\documentclass{article}
\usepackage[top=1in, bottom=1in, left=1in, right=1in]{geometry} 

\usepackage[cmex10]{amsmath} 
\allowdisplaybreaks
\usepackage{cite} 
\usepackage{hyperref} 
\usepackage{xcolor}
\definecolor{darkblue}{rgb}{.15,0,.7}
\hypersetup{
	colorlinks=true,
	linkcolor=darkblue,
	citecolor=darkblue,
	urlcolor=darkblue
}

\usepackage{amssymb} 
\usepackage[T1]{fontenc}
\usepackage{textcomp}
\usepackage{bm} 

\usepackage[pdftex]{graphicx} 
\usepackage{subcaption}
\usepackage{amsthm} 
\newtheorem{theorem}{Theorem}
\newtheorem*{theorem*}{Theorem}
\newtheorem{lemma}{Lemma}
\newtheorem*{lemma*}{Lemma}
\newtheorem*{claim*}{Claim}

\theoremstyle{definition}
\newtheorem{definition}{Definition}

\theoremstyle{remark}

\usepackage{flushend}
\usepackage{enumitem}

\renewcommand{\Pr}{\mathbb{P}} 
\newcommand{\expect}{\mathbb{E}} 

\newcommand{\pmf}{pmf}
\newcommand{\naturalnumber}{\mathbb{N}} 

\newcommand{\Databf}{\bm{X}} 
\newcommand{\Data}{X} 
\newcommand{\databf}{\bm{x}} 
\newcommand{\data}{x}
\newcommand{\Signalbf}{\bm{S}}
\newcommand{\Signal}{S}
\newcommand{\signalbf}{\bm{s}}
\newcommand{\signal}{s}
\newcommand{\cost}{g} 
\newcommand{\nonp}{\bot} 
\newcommand{\paymecgenbf}{\bm{R}}
\newcommand{\paymecgen}{R}
\newcommand{\paymecbf}{\bm{R}^{(N,\epsilon)}}
\newcommand{\paymec}{R^{(N,\epsilon)}}
\newcommand{\quality}{\theta} 
\newcommand{\State}{W}
\newcommand{\state}{w}
\newcommand{\Value}{V}
\newcommand{\independent}{\protect\mathpalette{\protect\independenT}{\perp}}
\def\independenT#1#2{\mathord{\rlap{$#1#2$}\mkern3mu{#1#2}}} 
\newcommand{\err}{\tau} 

\makeatletter
\newcommand*{\defeq}{\mathrel{\vcenter{\baselineskip0.5ex \lineskiplimit0pt
                     \hbox{\scriptsize.}\hbox{\scriptsize.}}}%
                     =}

\renewcommand{\pmf}{PMF}
\newcommand{\player}{individual}
\newcommand{\players}{individuals}
\newcommand{\Player}{Individual}
\newcommand{\Players}{Individuals}
\newcommand{\principal}{data collector}

\begin{document}
\title{The Value of Privacy: Strategic Data Subjects, Incentive Mechanisms and Fundamental Limits}

\author{
Weina Wang, Lei Ying and Junshan Zhang\\
	{School of Electrical, Computer and Energy Engineering}\\
	{Arizona State University}\\
	{Tempe, AZ 85287}\\
	email: \{weina.wang, lei.ying.2, junshan.zhang\}@asu.edu
}

\maketitle

\begin{abstract}
We study the value of data privacy in a game-theoretic model of trading private data, where a data collector purchases private data from strategic data subjects (\players{}) through an incentive mechanism. The private data of each \player{} represents her knowledge about an underlying state, which is the information that the \principal{} desires to learn. Different from most of the existing work on privacy-aware surveys, our model does not assume the \principal{} to be trustworthy. Then, an \player{} takes full control of its own data privacy and reports only a privacy-preserving version of her data.

In this paper, the value of $\epsilon$ units of privacy is measured by the minimum payment of all nonnegative payment mechanisms, under which an \player{}'s best response at a Nash equilibrium is to report the data with a privacy level of $\epsilon$. The higher $\epsilon$ is, the less private the reported data is. We derive lower and upper bounds on the value of privacy which are asymptotically tight as the number of data subjects becomes large. Specifically, the lower bound assures that it is impossible to use less amount of payment to buy $\epsilon$ units of privacy, and the upper bound is given by an achievable payment mechanism that we designed. Based on these fundamental limits, we further derive lower and upper bounds on the minimum total payment for the \principal{} to achieve a given learning accuracy target, and show that the total payment of the designed mechanism is at most one \player{}'s payment away from the minimum.


\end{abstract}

\section{Introduction}
From the monetary coupons offered for revealing opinions of a product to the large-scale trade of personal information by data brokers such as Acxiom \cite{Kro_14}, the commoditization of private data has been trending up when big data analytics is playing a more and more critical role in advertising, scientific research, etc. However, in the wake of a number of recent scandals, such as the Netflix data breach and the Veterans Affairs data theft, data privacy is emerging as one of the most serious concerns of big data analytics. This raises a fundamental question ``whether big-data and privacy can go hand-by-hand or giving up our privacy is inevitable in the big-data era.'' One common practice of collecting private data is called informed consent. With information on ``who is collecting the data, what data is collected, and how the data will be used,'' data subjects decide upon whether to report data or not. The data collector is supposed to use the data only in the manner disclosed to data subjects. This practice, however, has two fundamental issues: (i) data subjects have no control of data privacy after transferring private data to the data collector; and (ii) the data collector has to take full responsibility of protecting users' private data, which not only costs significant investment on infrastructure and maintenance, but also may lead to reputation damage if data breach occurs. In some applications, such as collecting certain browsing history records to enhance the phishing and malware protection of web browsers \cite{ErlPihKor_14,FanPihErl_15}, the data collectors prefer to avoid holding individuals' raw data for subpoena concerns.

Taking a forward-looking view, we envisage a market model for private data analytics such that private data is treated as a commodity and traded in the market. In particular, the data collector will use an incentive mechanism to pay (or reward) individuals for reporting informative data, and individuals control their own data privacy by reporting noisy data with the appropriate level of privacy protection (or level of noise added) being strategically chosen to maximize their payoffs. A distinctive merit of this privacy protection approach is that data subjects take full control of their own privacy and the data collector gets informative data but does not need to bear the responsibility of protecting data privacy. This differentiates our approach from the existing work \cite{GhoRot_11,FleLyu_12,LigRot_12,RotSch_12,GhoLig_13,NisVadXia_14,GhoLigRot_14}, where the data collector is assumed to be a trustworthy entity who is willing to and has the capability to protect users' privacy.

One significant challenge of the proposed paradigm is that the data collector has no direct control (perhaps no information either) over the quality of reported data. To tackle this challenge, we cast the problem into a game-theoretic setting, which allows us to quantify two fundamental tradeoffs: the tradeoff between cost and accuracy from the data collector's perspective, and the tradeoff between reward and privacy from an individual's perspective (the value of privacy for a data subject). In return, with the reward (incentive) as the bridge, it establishes the tradeoff of data privacy concerned by an individual versus data quality concerned by the data collector.

Specifically, we consider a game-theoretic model of collecting private data in hypothesis testing, where the data collector is interested in learning information from a population of $N$ individuals. An illustration of our model is shown in Figure~\ref{fig:model}. The information is represented by a binary random variable $\State$, which is called the \emph{state}. Each \player{}~$i$ possesses a binary \emph{signal} $\Signal_i$, which is her private data, representing her knowledge about the state $\State$. Conditional on the state $\State$, the signals are independently generated such that the probability for each signal $\Signal_i$ to be the same as $\State$ is $\quality$, where $0.5<\theta<1$. To protect her privacy, an \player{} reports only a privacy-preserving version of her signal, denoted by $X_i$, or chooses to not participate after considering both the payment from the \principal{} and the loss of privacy. The \principal{} needs to decide the amount of payment and the payment mechanism to get informative reports, i.e., not completely random data. Intuitively, the higher the payment is, the more informative the reported data should be. We will answer the following fundamental questions in this paper: \emph{What is the minimum payment needed from the \principal{} to obtain reported data with a privacy level $\epsilon$?} \emph{Which payment mechanism can be used to collect private data with minimum cost?} This setting without accounting for data privacy has garnered much attention in the literature (see, e.g., \cite{MilResZec_09,AceDahLob_11,LeSubBer_14}), including the application of estimating the underlying value of a new technology by eliciting opinions from individuals.

Intuitively, the \principal{} can purchase more informative data (so higher privacy) by offering higher payment. However, the strategic behavior of the privacy-aware \players{} makes this more complicated. Due to privacy concerns, an \player{}'s action\slash strategy is the conditional distributions of the reported data given the realizations of the signal. But the actions of the \players{} are not observable to the \principal{}. Instead, what the \principal{} receives is the reported data, generated randomly according to the \players{}' strategies, so the payments can only be designed based on the reported data. This differs our problem from the conventional mechanism design.

Furthermore, the privacy-aware \players{} weigh the privacy loss against the payment to choose the best quantity of privacy to trade. To make an \player{} willing to trade $\epsilon$ level of privacy, the \principal{} needs to make sure doing this benefits the \player{} most. We reiterate that the \principal{} has access only to the reported data instead of the \players{}' actions. Note that only compensating the privacy cost incurred is not sufficient. The payment mechanism needs to ensure that $\epsilon$ is the best privacy level such that when an \player{} uses a less-private strategy, the decrease in her payment is faster than the decrease in her privacy cost, and similarly, when an \player{} uses a more-private strategy, the increase in her payment is slower than the increase in her privacy cost. In other words, with a game-theoretic approach, we consider an \player{}'s best response in a Nash equilibrium, and the value of data privacy is measured by the minimum payment that makes this equilibrium strategy have a privacy level of $\epsilon$, which represents the monetary value of data privacy in a market for private data.

\subsection*{Summary of Main Results}
It is assumed that \players{} use the celebrated notion of differential privacy \cite{DwoMcSNis_06,Dwo_06} to evaluate their data privacy. When an \player{}~$i$ uses an $\epsilon$-differentially private randomization strategy to generate $\Data_i$, the privacy loss incurred is $\epsilon$, and the \player{}'s cost of privacy loss is a function of $\epsilon$, whose form is assumed to be publicly known. The value of $\epsilon$ units of privacy, denoted by $\Value(\epsilon)$, is measured by the minimum payment of all nonnegative payment mechanisms under which an \player{}'s best response in a Nash equilibrium is to report the data with privacy level $\epsilon$, where nonnegativity ensures that \players{} would not be \emph{charged} for reporting data. We are interested in the range that $\epsilon>0$, simply because when $\epsilon=0$, the reported data is independent of the private data and thus would be of no use for data analysis. Our contributions are summarized as follows:
\begin{enumerate}[leftmargin=1.3em,topsep=1ex,itemsep=0ex]
\item We establish a lower bound on $\Value(\epsilon)$. First we characterize the strategies of \players{} at a Nash equilibrium to prove that from a payment perspective, it suffices to focus on nonnegative payment mechanisms at which the best response of an \player{} in a Nash equilibrium is a symmetric randomized response with a privacy level of $\epsilon$. This strategy generates the reported data by flipping the signal with probability $\frac{1}{e^{\epsilon}+1}$: for convenience, this is called the $\epsilon$-strategy. Next we prove that the expected payments resulting from any Nash equilibrium of any payment mechanism can be ``replicated'' by a genie-aided payment mechanism, where the payments are determined with the aid of a genie who knows the underlying state $\State$. This makes the analysis of the Nash equilibria more tractable by decoupling the \players{} in the payments. The lower bound is then given by necessary conditions for $\epsilon$ to be the best privacy level in the genie-aided mechanism. We remark that although the genie-aided mechanism that achieves the lower bound is not implementable, it can be well-approximated, when the number of \players{} is large, by the feasible payment mechanism that we design to establish the upper bound.

\begin{sloppypar}
\item We observe that the equilibrium strategies exhibit some interesting characteristics: the strategy of an \player{} in a Nash equilibrium is either a symmetric randomized response, which treats the realizations of the private signal symmetrically, or a non-informative strategy, where the reported data is independent of the signal. This characterization holds regardless of the prior distribution of the state, and it also holds for more general probability models of the signals. This characterization advances our understanding of the behavior of privacy-aware \players{}. It is worth pointing out that finding an equilibrium strategy of a privacy-aware \player{} under some payment mechanism involves non-convex optimization.
\end{sloppypar}

\item We prove an upper bound on $\Value(\epsilon)$ by designing a payment mechanism $\paymecbf$, in which the strategy profile consisting of $\epsilon$-strategies constitutes a Nash equilibrium. The expected payment to each \player{} at this equilibrium gives an upper bound on $\Value(\epsilon)$. This upper bound converges to the lower bound exponentially fast as the number of \players{} $N$ becomes large, which indicates that the lower and upper bounds are asymptotically tight.

\item The above fundamental bounds on the value of privacy can be further used to study the \emph{payment--accuracy problem}, where the \principal{} aims to minimize the total payment while achieving an accuracy target in learning the state $\State$. Given an accuracy target $\err$, which can be regarded as the maximum allowable error, let $F(\err)$ denote the minimum total payment for achieving $\err$. We obtain lower and upper bounds on $F(\err)$ based on the lower and upper bounds on the value of privacy. The upper bound is given by the designed mechanism $\paymecbf$ with properly chosen parameters, which shows that the total payment of the designed mechanism is at most one \player{}'s payment away from the minimum.
\end{enumerate}

\section{Related Work}\label{sec:related}
Most existing work on privacy-aware surveys \cite{GhoRot_11,FleLyu_12,LigRot_12,RotSch_12,GhoLig_13,NisVadXia_14,GhoLigRot_14} assumes that there is a trusted data curator or data collector. The private data is either already kept by the \principal{}, or is elicited using mechanisms that are designed with the aim of truthfulness. What the \principal{} purchases is the ``right'' of using individuals' data in an announced way. Our work differs from the existing work by considering a \principal{} who is not trusted by individuals. In this scenario, the \principal{} directly purchases the private data, in which privacy is embedded.

In the seminal work by Ghosh and Roth \cite{GhoRot_11}, individuals' data is already known to the \principal{}, and individuals bid their costs of privacy loss caused by data usage, where each individual's privacy cost is modeled as a linear function of $\epsilon$ if her data is used in an $\epsilon$-differentially private manner. The goal of the mechanism design is to elicit truthful bids of individuals' cost functions, i.e., the coefficients. Subsequent work \cite{FleLyu_12,LigRot_12,RotSch_12,NisVadXia_14} explores various models for individuals' valuation of privacy, especially the correlation between the coefficients and the private bits.

This line of work has been extended to the scenario that the data is not available yet and needs to be reported by the individuals to the \principal{}, but the \principal{} is still trusted \cite{GhoLig_13,Xia_13,CheChoKas_13,GhoLigRot_14}. Notably, Ghosh, Ligett and Roth \cite{GhoLigRot_14} study the model in which the collected data is non-verifiable. The goal of the mechanism design there is to incentivize truthful data reporting (without adding any noise) from individuals. For more work on the interplay between differential privacy and mechanism design, Pai and Roth \cite{PaiRot_13} give a comprehensive survey.

The local model of differential privacy, which is a generalization of randomized response \cite{War_65} and is formalized in \cite{KasLeeNis_11}, has been studied in the literature \cite{DwoMcSNis_06,Dwo_06,HsuKhaRo_12,DucJorWai_13,DwoRot_14,KaiOhVis_14,WanYinZha_14,WanYinZha_15,BasSmi_15,Sho_15}. The hypothesis testing formulation in our paper is similar to a setting in \cite{KaiOhVis_14}, where the authors find an optimal mechanism that maximizes the statistical discrimination of the hypotheses subject to local differential privacy constraint. In practice, Google's Chrome web browser has implemented the RAPPOR mechanism \cite{ErlPihKor_14,FanPihErl_15} to collect users' data, which guarantees that only limited privacy will be leaked by using randomized response in a novel manner. However, users may still not be willing to report data in the desired way due to the lack of an incentive mechanism.

\section{System Model}\label{sec:model}
We consider a single-bit learning problem with privacy-aware \players{} as shown in Figure~\ref{fig:model}. Recall that the \principal{} is interested in learning the state $\State$, which is a binary random variable. For example, the state $\State$ can describe the underlying value of some new technology. Let $P_\State$ denote the prior \pmf\ of $\State$. We assume that $P_\State(1)>0$ and $P_\State(0)>0$.
\begin{figure}[t]
\centering
\includegraphics[scale=0.25]{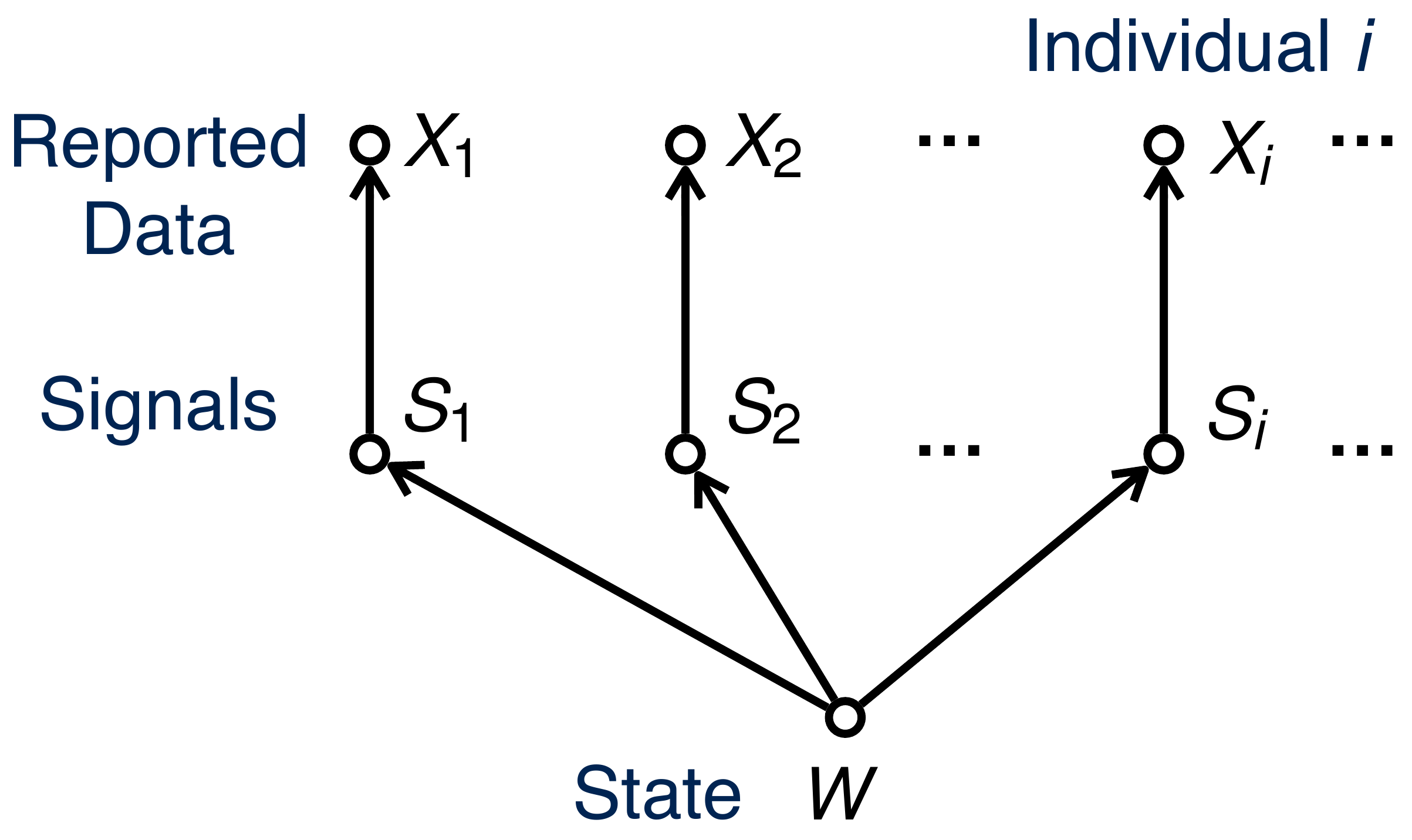}
\caption{Information structure of the model: The \principal{} is interested in the state $\State$, which is a binary random variable. Each \player{}~$i$ possesses her private data, which is a binary signal $\Signal_i$. Conditioned on $\State$, $S_1,S_2,\dots,S_N$ are i.i.d.\ \Player{}~$i$'s reported data is $\Data_i$, which is generated based on $\Signal_i$ using a randomized strategy.}
\label{fig:model}
\end{figure}

\textbf{\Players{} and Strategies.} Consider a population of $N$ \players{} and denote the set of \players{} by $\mathcal{N}=\{1,2,\dots,N\}$. Denote all \players{} other than some given \player{}~$i$ by ``$-i$.'' Each \player{}~$i$ possesses a binary signal $\Signal_i$, which is her private data, reflecting her knowledge about the state $\State$. For example, $\Signal_i$ can represent \player{}~$i$'s opinion towards the new technology. Let $\Signalbf=(\Signal_1,\Signal_2,\cdots,\Signal_N)$. Conditional on the state $\State$, the signals $\Signal_1,\Signal_2,\dots,\Signal_N$ are i.i.d.\ with the following conditional distributions:
\begin{align*}
\Pr(\Signal_i=1\mid \State=1)=\quality,\mspace{18mu}&\Pr(\Signal_i=0\mid \State=1)=1-\quality,\\
\Pr(\Signal_i=0\mid \State=0)=\quality,\mspace{18mu}&\Pr(\Signal_i=1\mid \State=0)=1-\quality,
\end{align*}
where the parameter $\quality$ with $0.5<\quality<1$ is called \emph{the quality of signals}.

Let $\Data_i$ denote the data reported by \player{}~$i$ and let $\Databf=(\Data_1,\Data_2,\dots,\Data_N)$. The acceptable values for reported data are $0$, $1$, and ``nonparticipation.'' So $\Data_i$ takes values in the set $\mathcal{\Data}=\{0,1,\nonp\}$, where $\nonp$ indicates that \player{}~$i$ declines to participate. A strategy of \player{}~$i$ for data reporting is a mapping $\sigma_i\colon\{0,1\}\rightarrow\mathcal{D}(\mathcal{\Data})$, where $\mathcal{D}(\mathcal{\Data})$ is the set of probability distributions on $\mathcal{\Data}$. Let $\bm{\sigma}=(\sigma_1,\sigma_2,\dots,\sigma_N)$. The strategy $\sigma_i$ prescribes a distribution to $\Data_i$ for each possible value of $\Signal_i$, which defines the conditional distribution of $\Data_i$ given $\Signal_i$. Since we will discuss different strategies of \player{}~$i$, we let $\Pr_{\sigma_i}(\Data_i=\data_i\mid \Signal_i=\signal_i)$ with $\data_i\in\mathcal{\Data}$ and $\signal_i\in\{0,1\}$ denote the conditional probabilities defined by strategy $\sigma_i$. If a strategy $\sigma_i$ satisfies that $\Pr_{\sigma_i}(\Data_i=1\mid \Signal_i=1)=\Pr_{\sigma_i}(\Data_i=0\mid \Signal_i=0)$ and $\Pr_{\sigma_i}(\Data_i=\nonp\mid \Signal_i=1)=\Pr_{\sigma_i}(\Data_i=\nonp\mid \Signal_i=0)=0$, we say $\sigma_i$ is a \emph{symmetric randomized response}. If a strategy $\sigma_i$ makes $\Data_i$ and $\Signal_i$ independent, we say $\sigma_i$ is \emph{non-informative}; otherwise we say $\sigma_i$ is \emph{informative}.

\textbf{Mechanism.} The \principal{} uses a payment mechanism $\paymecgenbf\colon\mathcal{\Data}^N\rightarrow \mathbb{R}^N$ to determine the amount of payment to each \player{}, where $\paymecgen_i(\databf)$ is the payment to \player{}~$i$ when the reported data is $\Databf=\databf$. We are interested in payment mechanisms in which the payment to each \player{} is nonnegative, i.e., $\paymecgen_i(\databf)\ge 0$ for any \player{}~$i$ and any $\databf\in\mathcal{\Data}^N$, which we call \emph{nonnegative mechanisms}. This constraint is motivated by the fact that in many practical applications such as surveys, the \principal{} has no means to charge users and can only use payments to incentivize user participation.

\textbf{Privacy Cost.} We quantify the privacy loss incurred when a strategy is in use by the level of (local) differential privacy \cite{DwoMcSNis_06,Dwo_06,KasLeeNis_11,DwoRot_14} of the strategy, defined as follows.

\begin{definition}
The level of (local) differential privacy, or simply the privacy level, of a strategy $\sigma_i$, denoted by $\zeta(\sigma_i)$, is defined to be
\begin{equation*}
\zeta(\sigma_i)=\max\biggl\{\ln\biggl(\frac{\Pr_{\sigma_i}(\Data_i\in\mathcal{E}\mid \Signal_i=\signal_i)}{\Pr_{\sigma_i}(\Data_i\in\mathcal{E}\mid \Signal_i=1-\signal_i)}\biggr)\colon
\mathcal{E}\subseteq \{0,1,\nonp\},\signal_i\in\{0,1\}\biggr\},
\end{equation*}
where we follow the convention that $0\slash 0=1$, and the strategy $\sigma_i$ is said to be $\zeta(\sigma_i)$-differentially private.
\end{definition}

The level of differential privacy quantifies the indistinguishablity between the conditional distributions of the reported data given different values of the signal, therefore measuring how disclosive the strategy is. Note that the amount of privacy leakage quantified by differential privacy is ``in addition'' to what the adversaries already know. We refer the reader to \cite{DwoMcSNis_06} for more semantic implications of differential privacy.

The privacy loss causes a cost to an \player{}. We assume that when using strategies with the same privacy level, \players{} experience the same cost of privacy. Thus, we model each \player{}'s cost of privacy by a function $\cost$ of the privacy level. We call $\cost$ the \emph{cost function} and the cost the \emph{privacy cost}. Our results can be extended to the case where the cost functions are heterogeneous (see the discussion in Section~\ref{sec:hete}). We assume that the form of $\cost$ is publicly known (Ghosh and Roth \cite{GhoRot_11} and subsequent work study the scenario that cost functions are private and design truthful mechanisms to elicit them).

We say the cost function $\cost$ is \emph{proper} if it satisfies the following three conditions:
\begin{align}
&\cost(\xi)\ge 0,\quad\forall \xi\ge 0,\label{eq:proper1}\\
&\cost(0)=0,\label{eq:proper2}\\
&\cost\text{ is non-decreasing},\label{eq:proper3}
\end{align}
where \eqref{eq:proper1} follows from the fact that a privacy cost is nonnegative, \eqref{eq:proper2} indicates that the privacy cost is $0$ when the reported data is independent of the private data, and \eqref{eq:proper3} means that the privacy cost will not decrease when the privacy loss becomes larger. In this paper, we will focus on a proper cost function that is convex, continuously differentiable, and $\cost(\xi)=0$ only for $\xi=0$. With a little abuse of notation, we also use $\cost(\sigma_i)$ to denote $\cost(\zeta(\sigma_i))$, which is the privacy cost to \player{}~$i$ when the strategy $\sigma_i$ is used.

\textbf{Game Formulation and Nash Equilibrium.} In this market model, the \principal{} first announces a payment mechanism. Then this mechanism induces a strategic form game where the \players{} are the players. The utility of each \player{} is the difference between her payment and her privacy cost. We assume that the \players{} are risk neutral, i.e., they are interested in maximizing their expected utility. In this game, the prior distribution $P_{\State}$, the signal quality parameter $\quality$, the form of the payment mechanism $\paymecgenbf$ and the cost function $\cost$ are common knowledge.

We focus on Nash equilibria of a payment mechanism, where each \player{} has no incentive to unilaterally change her strategy given other \players{}' strategies. Formally, a Nash equilibrium in our model is defined as follows.
\begin{definition}
A strategy profile $\bm{\sigma}$ is a Nash equilibrium in a payment mechanism $\paymecgenbf$ if for any \player{}~$i$ and any strategy $\sigma_i'$,
\begin{equation*}
\expect_{\bm{\sigma}}[\paymecgen_i(\Databf)-\cost(\sigma_i)]\ge\expect_{(\sigma_i',\bm{\sigma}_{-i})}[\paymecgen_i(\Databf)-\cost(\sigma_i')],
\end{equation*}
where the expectation is over the reported data $\Databf$, and the subscripts $\bm{\sigma}$ and $(\sigma_i',\bm{\sigma}_{-i})$ indicate that $\Databf$ is generated by the strategy profile $\bm{\sigma}$ and $(\sigma_i',\bm{\sigma}_{-i})$, respectively.
\end{definition}

\section{The Value of Data Privacy}
We say that the \principal{} \emph{obtains} $\epsilon$ units of privacy from an \player{}~$i$ in a payment mechanism if \player{}~$i$'s best response in a Nash equilibrium of the mechanism is to report data with a privacy level of $\epsilon$. Recall that we are interested in the regime $\epsilon>0$ since the \principal{} wants the reported data to be useful for data analysis. Let $\mathcal{\paymecgenbf}(i;\epsilon)$ denote the set of nonnegative payment mechanisms in which the \principal{} obtains $\epsilon$ units of privacy from \player{}~$i$. Then we measure the value of $\epsilon$ units of privacy by the minimum payment to \player{}~$i$ of all mechanisms in $\mathcal{\paymecgenbf}(i;\epsilon)$. Note that this measure does not depend on the specific identity of $i$ due to the symmetry across \players{}. For any mechanism $\paymecgenbf\in\mathcal{\paymecgenbf}(i;\epsilon)$, let $\bm{\sigma}^{(\paymecgenbf;\epsilon)}$ denote the corresponding Nash equilibrium. Then, formally, the value of $\epsilon$ units of privacy is measured by
\begin{equation}\label{eq:value}
\Value(\epsilon)=\inf_{\paymecgenbf\in\mathcal{\paymecgenbf}{(i;\epsilon)}}\expect_{\bm{\sigma}^{(\paymecgenbf;\epsilon)}}[\paymecgen_i(\Databf)].
\end{equation}

In this section, we first derive a lower bound on $\Value(\epsilon)$ by characterizing the Nash equilibria and replicating mechanisms in $\mathcal{\paymecgenbf}(i;\epsilon)$ by genie-aided mechanisms. We then design a payment mechanism in $\mathcal{\paymecgenbf}(i;\epsilon)$, and consequently the equilibrium payment to \player{}~$i$ in this mechanism serves as an upper bound of $\Value(\epsilon)$. The gap between the lower and upper bounds diminishes to zero exponentially fast as the number of \players{} becomes large, which indicates that the lower and upper bounds are asymptotically tight.

\subsection{Lower Bound}\label{sec:lower}
We present a lower bound on $\Value(\epsilon)$ in Theorem~\ref{thm:lower} below. For convenience, we define
\begin{equation}\label{eq:lowerbound}
\Value_{\mathrm{LB}}(\epsilon)=\cost'(\epsilon)\frac{e^{\epsilon}+1}{e^{\epsilon}}\biggl(\frac{\quality}{2\quality-1}(e^{\epsilon}+1)-1\biggr),
\end{equation}
where $\cost'$ is the derivative of the privacy cost function of an \player{} and $\quality$ is the quality of signals.
\begin{theorem}\label{thm:lower}
\begin{sloppypar}
The value of $\epsilon$ units of privacy measured in \eqref{eq:value} for any $\epsilon>0$ is lower bounded as $\Value(\epsilon)\ge\Value_{\mathrm{LB}}(\epsilon)$. Specifically, for any nonnegative payment mechanism $\paymecgenbf$, if the strategy of an \player{}~$i$ in a Nash equilibrium has a privacy level of $\epsilon$ with $\epsilon>0$, then the expected payment to \player{}~$i$ at this equilibrium is lower bounded by $\Value_{\mathrm{LB}}(\epsilon)$.
\end{sloppypar}
\end{theorem}

We remark that the lower bound in Theorem~\ref{thm:lower} can be achieved by a hypothetical payment mechanism in which a genie who knows the realization of the underlying state $\State$ guides the \principal{} on how much to pay each \player{}. Intuitively, the knowledge of the state $\State$ provides more information about the system, which helps the \principal{} to obtain privacy with less payment. While it may sound like a chicken-and-egg problem as the \principal{}'s sole purpose of paying \players{} for their private data is to learn the state $\State$, it will become clear that the philosophy carries over and the \principal{} should utilize the best estimate of $\State$ in the payment mechanism to minimize the payment. The insight we gain from this mechanism sheds light on the asymptotically tight upper bound on the value of privacy in Section~\ref{sec:upper}.

This genie-aided payment mechanism, denoted by $\widehat{\paymecgenbf}^{(\epsilon)}$, determines the payment to each \player{}~$i$ based on her own reported data $\Data_i$ and the state $\State$ as follows:
\begin{equation}\label{eq:hypo-mec}
\widehat{\paymecgen}^{(\epsilon)}_i(\Data_i,\State)=\frac{\cost'(\epsilon)(e^{\epsilon}+1)^2}{2e^{\epsilon}}\widehat{A}_{\Data_i,\State},
\end{equation}
where
\begin{gather*}
\widehat{A}_{1,1}=\frac{1}{(2\quality-1)P_{\State}(1)},\quad \widehat{A}_{0,0}=\frac{1}{(2\quality-1)P_{\State}(0)},\\
\widehat{A}_{0,1}=\widehat{A}_{1,0}=0.
\end{gather*}
In this mechanism, it can be proved that the best response of individual~$i$ is the following symmetric randomized response, denoted by $\sigma_i^{(\epsilon)}$, which is $\epsilon$-differentially private:
\begin{equation*}
\begin{split}
&\Pr_{\sigma^{(\epsilon)}_i}(\Data_i=1\mid \Signal_i=1)=\Pr_{\sigma^{(\epsilon)}_i}(\Data_i=0\mid \Signal_i=0)=\frac{e^{\epsilon}}{e^{\epsilon}+1},\\
&\Pr_{\sigma^{(\epsilon)}_i}(\Data_i=1\mid \Signal_i=0)=\Pr_{\sigma^{(\epsilon)}_i}(\Data_i=0\mid \Signal_i=1)=\frac{1}{e^{\epsilon}+1},\\
&\Pr_{\sigma^{(\epsilon)}_i}(\Data_i=\nonp\mid \Signal_i=1)=\Pr_{\sigma^{(\epsilon)}_i}(\Data_i=\nonp\mid \Signal_i=0)=0.
\end{split}
\end{equation*}
For convenience, we will refer to this strategy as the \emph{$\epsilon$-strategy}. The expected payment to \player{}~$i$ at this strategy equals to the lower bound in Theorem~\ref{thm:lower}.

Next we sketch the proof of Theorem~\ref{thm:lower}. We first give three lemmas that form the basis of the proof, and then present the proof based on that. The proofs of these lemmas are presented in Appendix~\ref{app:lem:equi-char}--\ref{app:lem:genie}.

\subsubsection{Characterization of Nash Equilibria} We first characterize \players{}' behavior in a Nash equilibrium. In general, an $\epsilon$-differentially private strategy has uncountably many possible forms. However, provided that the strategy is part of a Nash equilibrium (i.e., a best response of an \player{}), the following lemma substantially reduces the space of possibilities. We remark that a similar phenomenon for privacy-aware \players{} has been observed in \cite{CheSheVad_14} in a different setting.

\begin{lemma}\label{lem:equi-char}
In any nonnegative payment mechanism, an \player{}'s strategy in a Nash equilibrium is either a symmetric randomized response, or a non-informative strategy.
\end{lemma}

We remark that Lemma~\ref{lem:equi-char} holds for more general probability models of the signals. The proof carries over as long as the support of the joint distribution of the signals is the entire domain $\{0,1\}^N$.

By Lemma~\ref{lem:equi-char}, if an \player{}'s strategy in a Nash equilibrium has a privacy level of $\epsilon$, where $\epsilon>0$, this equilibrium strategy is either the $\epsilon$-strategy or the $(-\epsilon)$-strategy. The following lemma says that from the payment perspective, it suffices to further focus on the case that it is the $\epsilon$-strategy.
\begin{lemma}\label{lem:positive-eps}
For any nonnegative payment mechanism $\paymecgenbf$ in which the strategy profile $(\sigma_i^{(-\epsilon)},\bm{\sigma}_{-i})$ with some $\epsilon>0$ is a Nash equilibrium, there exists another nonnegative payment mechanism $\paymecgenbf'$ in which $(\sigma_i^{(\epsilon)},\bm{\sigma}_{-i})$ is a Nash equilibrium, and the expected payment to each \player{} at these two equilibria of the two mechanisms are the same.
\end{lemma}

This lemma is proved by considering the payment mechanism $\paymecgenbf'$ that is constructed by applying $\paymecgenbf$ on the reported data after modifying $\Data_i$ to $1-\Data_i$.

\subsubsection{Genie-Aided Payment Mechanism}
A genie-aided payment mechanism $\widehat{\paymecgenbf}\colon\mathcal{\Data}^N\times\{0,1\}\rightarrow\mathbb{R}^N$ determines the payment to an \player{} based on not only the reported data $\Databf$ but also the underlying state $\State$. Compared with a standard payment mechanism, a genie-aided mechanism is hypothetical since the \principal{} has access to the underlying state, as if she were aided by a genie. We consider nonnegative genie-aided payment mechanisms where $\widehat{\paymecgen}_i(\Databf,\State)$, the payment to \player{}~$i$, depends on only her own reported data $\Data_i$ and the underlying state $\State$. We write $\widehat{\paymecgen}_i(\Data_i,\State)$ to represent $\widehat{\paymecgen}_i(\Databf,\State)$ for conciseness. Therefore, for each individual~$i$, a genie-aided mechanism makes use of the information of $\State$ but discards the information in $\Databf_{-i}$. The following lemma shows that the expected payments resulting from any Nash equilibrium of any payment mechanism can be replicated by a genie-aided payment mechanism with the same Nash equilibrium. Thus we can restrict our attention to genie-aided mechanisms to obtain a lower bound on the value of privacy.

\begin{lemma}\label{lem:genie}
For any nonnegative payment mechanism $\paymecgenbf$ and any Nash equilibrium $\bm{\sigma}$ of it, there exists a nonnegative genie-aided mechanism $\widehat{\paymecgenbf}$, such that $\bm{\sigma}$ is also a Nash equilibrium of $\widehat{\paymecgenbf}$ and the expected payment to each \player{} at this equilibrium is the same under $\paymecgenbf$ and $\widehat{\paymecgenbf}$.
\end{lemma}

This lemma is proved by constructing the following genie-aided payment mechanism $\widehat{\paymecgenbf}$ according to the desired equilibrium $\bm{\sigma}$: for any \player{}~$i$ and any $\data_i\in\mathcal{\Data},\state\in\{0,1\}$,
\begin{equation*}
\widehat{\paymecgen}_i(\data_i,\state)=\overline{\paymecgen}_i(\data_i;\state)\defeq\expect_{\bm{\sigma}}[\paymecgen_i(\Databf)\mid \Data_i=\data_i,\State=\state].
\end{equation*}
Our intuition is as follows. A genie-aided mechanism can use the state $\State$ to generate an incentive to \player{}~$i$, which ``mimics'' the incentive provided by the reported data $\Databf_{-i}$ of others. The above genie-aided payment mechanism $\widehat{\paymecgenbf}$ is constructed such that no matter what strategy \player{}~$i$ uses,
her expected utility is the same under $\paymecgenbf$ and $\widehat{\paymecgenbf}$. Since an \player{} calculates her best response according to the expected utility, her equilibrium behavior and expected payment are the same under $\widehat{\paymecgenbf}$ and $\paymecgenbf$. We remark that the Nash equilibria of a genie-aided mechanism are much easier to analyze since
the \players{} are decoupled in the payments and thus an \player{}'s strategy does not have an influence on other \players{}' utility.

Let $\widehat{\mathcal{\paymecgenbf}}(i;\epsilon)$ denote the set of nonnegative genie-aided payment mechanisms in which the $\epsilon$-strategy is an \player{}~$i$'s strategy in a Nash equilibrium, and let $\sigma_i^{(\epsilon)}$ denote the $\epsilon$-strategy. Consider
\begin{equation*}
\widehat{\Value}(\epsilon)=\inf_{\widehat{\paymecgenbf}\in\widehat{\mathcal{\paymecgenbf}}(i;\epsilon)}\expect_{\sigma_i^{(\epsilon)}}\Bigl[\widehat{\paymecgen}_i(\Data_i,\State)\Bigr],
\end{equation*}
which is a definition similar to the value of $\epsilon$ units of privacy, $\Value(\epsilon)$, measured in \eqref{eq:value}. Then $\widehat{\Value}(\epsilon)\le \Value(\epsilon)$ for the following reasons. Consider any $\paymecgenbf\in\mathcal{\paymecgenbf}(i;\epsilon)$, i.e., any nonnegative payment mechanism $\paymecgenbf$ in which \player{}~$i$'s strategy in a Nash equilibrium has a privacy level of $\epsilon$. With Lemma~\ref{lem:equi-char} and \ref{lem:positive-eps}, we can assume without loss of generality that this equilibrium strategy is the $\epsilon$-strategy. Then by Lemma~\ref{lem:genie}, we can map $\paymecgenbf$ to a $\widehat{\paymecgenbf}\in\widehat{\mathcal{\paymecgenbf}}(i;\epsilon)$, such that
\begin{equation*}
\expect_{\bm{\sigma}^{(\paymecgenbf;\epsilon)}}[\paymecgen_i(\Databf)]=\expect_{\sigma_i^{(\epsilon)}}\Bigl[\widehat{\paymecgen}_i(\Data_i,\State)\Bigr].
\end{equation*}
Therefore, the infimum over $\widehat{\mathcal{\paymecgenbf}}(i;\epsilon)$ is no greater than the infimum over $\mathcal{\paymecgenbf}(i;\epsilon)$, i.e., $\widehat{\Value}(\epsilon)\le \Value(\epsilon)$.

\subsubsection{Proof of Theorem~\ref{thm:lower}}
With Lemma~\ref{lem:equi-char}, \ref{lem:positive-eps} and \ref{lem:genie}, we can prove the lower bound in Theorem~\ref{thm:lower} by focusing on the genie-aided mechanisms in $\widehat{\mathcal{\paymecgenbf}}(i;\epsilon)$. Then there is no need to consider the strategies of \players{} other than \player{}~$i$ since a genie-aided mechanism pays \player{}~$i$ only according to $\Data_i$ and $\State$. A necessary condition for the $\epsilon$-strategy to be a best response of \player{}~$i$ is that $\epsilon$ yields no worse expected payment than other privacy levels. We utilize this necessary condition to obtain a lower bound on the expected payment to \player{}~$i$, which gives a lower bound on $\widehat{\Value}(\epsilon)$ and further proves the lower bound in Theorem~\ref{thm:lower}.

\begin{proof}[Proof of Theorem~\ref{thm:lower}]
By Lemma~\ref{lem:equi-char}, \ref{lem:positive-eps} and \ref{lem:genie}, it suffices to focus on nonnegative genie-aided payment mechanisms in which the $\epsilon$-strategy is an \player{}~$i$'s strategy in a Nash equilibrium, i.e., mechanisms in $\widehat{\mathcal{\paymecgenbf}}(i;\epsilon)$. Consider any $\widehat{\paymecgenbf}\in\widehat{\mathcal{\paymecgenbf}}(i;\epsilon)$ and denote the $\epsilon$-strategy by $\sigma_i^{(\epsilon)}$. Consider the $\xi$-strategy of \player{}~$i$ with any $\xi\ge 0$ and denote it by $\sigma_i^{(\xi)}$. Then the expected utility of \player{}~$i$ at the strategy $\sigma_i^{(\xi)}$ can be written as
\begin{align*}
&\mspace{23mu}\expect_{\sigma_i^{(\xi)}}\Bigl[\widehat{\paymecgen}_i(\Data_i,\State)\Bigr]-\cost(\sigma_i^{(\xi)})\\
&=\sum_{\data_i,\signal_i,\state}\Pr_{\sigma_i^{(\xi)}}(\Data_i=\data_i\mid \Signal_i=\signal_i)\Pr(\Signal_i=\signal_i,\State=\state)\widehat{\paymecgen}_i(\data_i,\state)\\
&\mspace{23mu}-\cost(\xi),\\
&=\overline{K}_1\frac{e^{\xi}}{e^{\xi}+1}+\overline{K}_0\frac{1}{e^{\xi}+1}+\overline{K}-\cost(\xi),
\end{align*}
where
\begin{align*}
\overline{K}_1&=\{\widehat{\paymecgen}_i(1,1)P_{\State}(1)\quality+\widehat{\paymecgen}_i(1,0)P_{\State}(0)(1-\quality)\}\\
&\mspace{23mu}-\{\widehat{\paymecgen}_i(0,1)P_{\State}(1)\quality+\widehat{\paymecgen}_i(0,0)P_{\State}(0)(1-\quality)\},\\
\overline{K}_0&=\{\widehat{\paymecgen}_i(1,1)P_{\State}(1)(1-\quality)+\widehat{\paymecgen}_i(1,0)P_{\State}(0)\quality\}\\
&\mspace{23mu}-\{\widehat{\paymecgen}_i(0,1)P_{\State}(1)(1-\quality)+\widehat{\paymecgen}_i(0,0)P_{\State}(0)\quality\},\\
\overline{K}&=\widehat{\paymecgen}_i(0,1)P_{\State}(1)+\widehat{\paymecgen}_i(0,0)P_{\State}(0).
\end{align*}
It can be seen that $\overline{K}_1$, $\overline{K}_0$ and $\overline{K}$ do not depend on $\xi$. Let this expected utility define a function $f$ of $\xi$; i.e.,
\begin{equation*}
f(\xi)=\overline{K}_1\frac{e^{\xi}}{e^{\xi}+1}+\overline{K}_0\frac{1}{e^{\xi}+1}-\cost(\xi)+\overline{K}.
\end{equation*}
Then a necessary condition for the $\epsilon$-strategy to be an equilibrium strategy is that $\epsilon$ maximizes $f(\xi)$, which implies that $f'(\epsilon)=0$ since $\epsilon>0$. Since
\begin{equation*}
f'(\xi)=(\overline{K}_1-\overline{K}_0)\frac{e^{\xi}}{(e^{\xi}+1)^2}-\cost'(\xi),
\end{equation*}
setting $f'(\epsilon)=0$ yields that
\begin{equation}\label{eq:K1K2diff}
\overline{K}_1-\overline{K}_0=\cost'(\epsilon)\frac{(e^{\epsilon}+1)^2}{e^{\epsilon}}.
\end{equation}

Now we calculate the expected payment to \player{}~$i$ at the $\epsilon$-strategy:
\begin{align*}
\expect_{\sigma_i^{(\epsilon)}}\Bigl[\widehat{\paymecgen}_i(\Data_i,\State)\Bigr]=-(\overline{K}_1-\overline{K}_0)\frac{1}{e^{\epsilon}+1}+(\overline{K}_1+\overline{K}).
\end{align*}
By definition,
\begin{align*}
\overline{K}_1+\overline{K}&=\widehat{\paymecgen}_i(1,1)P_{\State}(1)\quality+\widehat{\paymecgen}_i(1,0)P_{\State}(0)(1-\quality)\\
&\mspace{23mu}+\widehat{\paymecgen}_i(0,1)P_{\State}(1)(1-\quality)+\widehat{\paymecgen}_i(0,0)P_{\State}(0)\quality,
\end{align*}
and
\begin{align*}
\overline{K}_1-\overline{K}_0&=\bigl(\widehat{\paymecgen}_i(1,1)-\widehat{\paymecgen}_i(0,1)\bigr)P_{\State}(1)(2\quality-1)\\
&\mspace{23mu}+\bigl(\widehat{\paymecgen}_i(0,0)-\widehat{\paymecgen}_i(1,0)\bigr)P_{\State}(0)(2\quality-1).
\end{align*}
Therefore,
\begin{align*}
\overline{K}_1+\overline{K}&=\frac{\quality}{2\quality-1}(\overline{K}_1-\overline{K}_0)\\
&\mspace{23mu}+\widehat{\paymecgen}_i(1,0)P_{\State}(0)+\widehat{\paymecgen}_i(0,1)P_{\State}(1)\\
&\ge\frac{\quality}{2\quality-1}(\overline{K}_1-\overline{K}_0)\\
&=\cost'(\epsilon) \frac{(e^{\epsilon}+1)^2}{e^{\epsilon}}\frac{\theta}{2\theta-1},
\end{align*}
where we have used the nonnegativity of $\widehat{\paymecgenbf}$. Then the expected payment to \player{}~$i$ is bounded as follows:
\begin{align}
&\mspace{23mu}\expect_{\sigma_i^{(\epsilon)}}\Bigl[\widehat{\paymecgen}_i(\Data_i,\State)\Bigr]\nonumber\\
&=-(\overline{K}_1-\overline{K}_0)\frac{1}{e^{\epsilon}+1}+(\overline{K}_1+\overline{K})\nonumber\\
&\ge\cost'(\epsilon) \frac{e^{\epsilon}+1}{e^{\epsilon}}\biggl(\frac{\quality}{2\quality-1}(e^{\epsilon}+1)-1\biggr),\label{eq:ineq-lower}
\end{align}
which proves the lower bound.
\end{proof}

Now beyond the proof, we take a moment to check when this lower bound can be achieved. To achieve the lower bound, we need the equality in \eqref{eq:ineq-lower} to hold and the equation \eqref{eq:K1K2diff} to be satisfied, which is equivalent to the following conditions:
\begin{gather}
\widehat{\paymecgen}_i(1,0)=0,\label{eq:hypo-mec10}\\
\widehat{\paymecgen}_i(0,1)=0,\label{eq:hypo-mec01}\\
(2\quality-1)\Bigl(\widehat{\paymecgen}_i(1,1)P_{\State}(1)+\widehat{\paymecgen}_i(0,0)P_{\State}(0)\Bigr)\nonumber\\
=\cost'(\epsilon)\frac{(e^{\epsilon}+1)^2}{e^{\epsilon}}.\label{eq:necessary-lower-gene}
\end{gather}
It is easy to check that the genie-aided payment mechanism $\widehat{\paymecgenbf}^{(\epsilon)}$ defined in \eqref{eq:hypo-mec} is in $\widehat{\mathcal{\paymecgenbf}}(i;\epsilon)$ and satisfies \eqref{eq:hypo-mec10}--\eqref{eq:necessary-lower-gene}, and therefore achieves the lower bound. Can this lower bound be achieved by a standard nonnegative payment mechanism? Consider any payment mechanism $\paymecgenbf\in\mathcal{\paymecgenbf}
(i;\epsilon)$. Following similar arguments, we can prove that to achieve the lower bound, $\paymecgenbf$ needs to satisfy the following conditions:
\begin{gather}
\overline{\paymecgen}_i(1;0)=0,\label{eq:necessary-lower1}\\
\overline{\paymecgen}_i(0;1)=0,\label{eq:necessary-lower2}\\
(2\quality-1)\bigl(\overline{\paymecgen}_i(1;1)P_{\State}(1)+\overline{\paymecgen}_i(0;0)P_{\State}(0)\bigr)\nonumber\\
=\cost'(\epsilon)\frac{(e^{\epsilon}+1)^2}{e^{\epsilon}},\label{eq:necessary-lower3}
\end{gather}
where recall that $\overline{\paymecgen}_i(\data_i;\state)=\expect_{\bm{\sigma}^{(\paymecgenbf,\epsilon)}}[\paymecgen_i(\Databf)\mid \Data_i=\data_i,\State=\state]$ for $\data_i,\state\in\{0,1\}$. It can be proved that if $\paymecgenbf$ satisfies \eqref{eq:necessary-lower1} and \eqref{eq:necessary-lower2}, then $\paymecgen_i(\databf)=0$ for any $\databf\in\mathcal{\Data}^N$, which contradicts \eqref{eq:necessary-lower3}. Therefore, no standard nonnegative payment mechanism can achieve the lower bound. However, as will be shown in the next section, we can design a class of standard nonnegative payment mechanisms such that the expected payment approaches the lower bound as the number of \players{} increases. The design follows the insights indicated by the genie-aided mechanism $\widehat{\paymecgenbf}^{(\epsilon)}$: to minimize the payment, the \principal{} should utilize the best estimate of $\State$ in the payment mechanism based on the noisy reports.

\subsection{Upper Bound}\label{sec:upper}
We present an upper bound on $\Value(\epsilon)$ in Theorem~\ref{thm:upper} below. For convenience, we define
\begin{equation}\label{eq:d}
d=\frac{1}{2}\ln\frac{(e^{\epsilon}+1)^2}{4(\quality e^{\epsilon}+1-\quality)((1-\quality)e^{\epsilon}+\quality)},
\end{equation}
where $\quality$ is the quality of signal. Note that $d>0$ for any $\epsilon>0$. Recall that $\Value_{\mathrm{LB}}(\epsilon)$ is the lower bound in Theorem~\ref{thm:lower}.

\begin{theorem}\label{thm:upper}
\begin{sloppypar}
The value of $\epsilon$ units of privacy measured in \eqref{eq:value} is upper bounded as $\Value(\epsilon)\le\Value_{\mathrm{LB}}(\epsilon)+O(e^{-Nd})$, where the $O(\cdot)$ is for $N\rightarrow\infty$. Specifically, there exists a nonnegative payment mechanism $\paymecbf$ in which the strategy profile $\bm{\sigma}^{(\epsilon)}$ consisting of $\epsilon$-strategies is a Nash equilibrium, and the expected payment to each \player{} $i$ at this equilibrium is upper bounded by $\Value_{\mathrm{LB}}(\epsilon)+O(e^{-Nd})$.
\end{sloppypar}
\end{theorem}
Comparing this upper bound with the lower bound $\Value_{\mathrm{LB}}(\epsilon)$ in Theorem~\ref{thm:lower} we can see that the gap between the lower and upper bounds is just the term $O(e^{-Nd})$, which diminishes to zero exponentially fast as $N$ goes to infinity.

We present the payment mechanism $\paymecbf$ in Section~\ref{sec:paymec}. We will show that under $\paymecbf$, the strategy profile $\bm{\sigma}^{(\epsilon)}$ consisting of $\epsilon$-strategies is a Nash equilibrium. Therefore, $\paymecbf$ is a member of $\mathcal{\paymecgenbf}(i;\epsilon)$, and the payment to \player{}~$i$ at $\bm{\sigma}^{(\epsilon)}$ gives an upper bound on the value of privacy.

The design of $\paymecbf$ is enlightened by the hypothetical payment mechanism $\widehat{\paymecgenbf}^{(\epsilon)}$ defined in \eqref{eq:hypo-mec}. But without direct access to the state $\State$, the mechanism $\paymecbf$ relies on the reported data from an \player{}~$i$'s peers, i.e., \players{} other than \player{}~$i$, to obtain an estimate of $\State$. We borrow the idea of the peer-prediction method \cite{MilResZec_09}, which rewards more for the agreement between an \player{} and her peers to encourage truthful reporting. However, unlike the peer-prediction method, the \players{} here have privacy concerns and they will weigh the privacy cost against the payment to choose the best privacy level. We modify the payments in $\widehat{\paymecgenbf}^{(\epsilon)}$ to ensure that the $\epsilon$-strategy is still a best response of each \player{} in $\paymecbf$, given that other \players{} also follow the $\epsilon$-strategy, which yields the desired Nash equilibrium $\bm{\sigma}^{(\epsilon)}$.

The equilibrium payment to each \player{} in $\paymecbf$ converges to the lower bound in Theorem~\ref{thm:lower} as the number of \players{} $N$ goes to infinity. The intuition behind is that as the number of \players{} $N$ goes to infinity, the majority of the reported data from other \players{} converges to the underlying state $\State$, and thus $\paymecbf$ works similar as the genie-aided mechanism $\widehat{\paymecgenbf}^{(\epsilon)}$, whose equilibrium payment to each \player{} equals to the lower bound in Theorem~\ref{thm:lower}.

\subsubsection{A Payment Mechanism \texorpdfstring{$\paymecbf$}{R(N,epsilon)}}\label{sec:paymec}
The payment mechanism $\paymecbf$ is designed for purchasing private data from $N$ privacy-aware \players{}, parameterized by a privacy parameter $\epsilon$, where $N\ge 2$ and $\epsilon>0$.
\begin{enumerate}[leftmargin=2em]
\item Each \player{} reports her data (which can be the decision of not participating).
\item The \principal{} counts the number of participants, which is denoted by $n$.
\item For non-participating \players{}, the payment is zero.
\item If there is only one participant, pay zero to this participant. Otherwise, for each participating \player{}~$i$, the \principal{} computes the variable
\begin{equation*}
M_{-i}=
\begin{cases}
1 & \text{if $\displaystyle\sum_{j\colon\Data_j\neq \perp,j\neq i}\Data_j\ge\Bigl\lfloor\frac{n-1}{2}\Bigr\rfloor + 1$,}\\
0 & \text{otherwise,}
\end{cases}
\end{equation*}
which is the majority of the other participants' reported data. Then the \principal{} pays \player{}~$i$ the following amount of payment according to her reported data $\Data_i$ and $M_{-i}$:
\begin{equation*}
\paymec_i(\Databf)=\frac{\cost'(\epsilon)(e^{\epsilon}+1)^2}{2e^{\epsilon}}A_{\Data_i,M_{-i}},
\end{equation*}
where the parameters $A_{1,1},A_{0,0},A_{0,1},A_{1,0}$ are defined in Section~\ref{sec:para}.
\end{enumerate}

\subsubsection{Payment Parameterization}\label{sec:para}
Let
\begin{equation*}
\alpha=\quality\frac{e^{\epsilon}}{e^{\epsilon}+1}+(1-\quality)\frac{1}{e^{\epsilon}+1}.
\end{equation*}
The physical meaning of $\alpha$ can be seen by considering the strategy profile $\bm{\sigma}^{(\epsilon)}$, where given the state $\State$, the reported data $\Data_1,\Data_2,\dots,\Data_N$ are i.i.d.\ with
\begin{equation*}
\Pr_{\sigma_i^{(\epsilon)}}(\Data_i=1\mid \State=1)=\Pr_{\sigma_i^{(\epsilon)}}(\Data_i=0\mid \State=0)=\alpha.
\end{equation*}
Given that the number of participants is $n$ with $n\ge 2$, define the following quantities. Consider a random variable that follows the binomial distribution with parameters $n-1$ and $\alpha$. Let $\beta^{(n)}$ denote the probability that this random variable is greater than or equal to $\lfloor\frac{n-1}{2}\rfloor + 1$. Let
\begin{equation}\label{eq:gamma}
\gamma^{(n)}=
\begin{cases}
\displaystyle
1- \binom{n-1}{\frac{n-1}{2}}\alpha^{\frac{n-1}{2}}(1-\alpha)^{\frac{n-1}{2}}& \text{if $n-1$ is even,}\\
1 & \text{if $n-1$ is odd.}
\end{cases}
\end{equation}
To see the physical meaning of $\beta^{(n)}$ and $\gamma^{(n)}$, still consider $\bm{\sigma}^{(\epsilon)}$, where the number of participants is $n=N$. Then for an \player{}~$i$,
\begin{align*}
\Pr_{\bm{\sigma}^{(\epsilon)}}(M_{-i}=1\mid \State=1)&=\beta^{(N)},\\
\Pr_{\bm{\sigma}^{(\epsilon)}}(M_{-i}=1\mid \State=0)&=\gamma^{(N)}-\beta^{(N)}.
\end{align*}

With the introduced notation, the parameters $A_{1,1}$, $A_{0,0}$, $A_{0,1}$, $A_{1,0}$ used in the payment mechanism $\paymecbf$ are defined as follows:
\begin{align*}
A_{1,1}&=\frac{P_\State(1)(1-\beta^{(n)})+P_\State(0)(1-(\gamma^{(n)}-\beta^{(n)}))}{(2\beta^{(n)}-\gamma^{(n)})(2\theta-1)P_\State(1)P_\State(0)},\\
A_{0,0}&=\frac{P_\State(1)\beta^{(n)}+P_\State(0)(\gamma^{(n)}-\beta^{(n)})}{(2\beta^{(n)}-\gamma^{(n)})(2\theta-1)P_\State(1)P_\State(0)},\\
A_{0,1}&=0,\\
A_{1,0}&=0.
\end{align*}
It is easy to verify that these parameters are nonnegative. Thus $\paymecbf$ is a nonnegative payment mechanism. The proof of the equilibrium properties of $\paymecbf$ in Theorem~\ref{thm:upper} is given below.

\subsubsection{Proof of Theorem~\ref{thm:upper}}
\begin{proof}
\begin{sloppypar}
It suffices to prove that the strategy profile $\bm{\sigma}^{(\epsilon)}$ is a Nash equilibrium in $\paymecbf$ and the expected payment to each \player{}~$i$ at this equilibrium satisfies that $\expect_{\bm{\sigma}^{(\epsilon)}}\Bigl[\paymec_i(\Databf)\Bigr]\le \Value_{\mathrm{LB}}(\epsilon)+O(e^{-Nd})$, where recall that $\Value_{\mathrm{LB}}(\epsilon)$ is defined in \eqref{eq:lowerbound}. For conciseness, in the remainder of this proof, we suppress the explicit dependence on $N$ and $\epsilon$, and write $\paymecgenbf$ and $\bm{\sigma}$ to represent $\paymecbf$ and $\bm{\sigma}^{(\epsilon)}$, respectively.
\end{sloppypar}

We first prove that the strategy profile $\bm{\sigma}$ is a Nash equilibrium in $\paymecgenbf$; i.e., for any \player{}~$i$, the $\epsilon$-strategy is a best response of \player{}~$i$ when other \players{} follow $\bm{\sigma}_{-i}$. Following the notation in the proof of Lemma~\ref{lem:equi-char}, for any \player{}~$i$ we consider any strategy $\sigma_i'$ of \player{}~$i$ and let
\begin{equation*}
\begin{split}
p_1=\Pr_{\sigma_i'}(\Data_i=1\mid \Signal_i=1),\quad& q_1=\Pr_{\sigma_i'}(\Data_i=0\mid \Signal_i=1),\\
p_0=\Pr_{\sigma_i'}(\Data_i=1\mid \Signal_i=0),\quad& q_0=\Pr_{\sigma_i'}(\Data_i=0\mid \Signal_i=0).
\end{split}
\end{equation*}
Then by the proof of Lemma~\ref{lem:equi-char}, the best response satisfies either $p_1=p_0,q_1=q_0$, or $p_1=q_0,p_0=q_1,p_1+q_1=1$, depending on the form of the utility function $U_i(p_1,p_0,q_1,q_0)$, which is the expected utility of \player{}~$i$ at the strategy $\sigma_i'$ when other \players{} follow $\bm{\sigma}_{-i}$. Thus, we derive the form of $U_i(p_1,p_0,q_1,q_0)$ next. Recall that we let $\overline{\paymecgen}_i(\data_i;\state)$ denote $\expect_{(\sigma_i',\bm{\sigma}_{-i})}[\paymecgen_i(\Databf)\mid \Data_i=\data_i, \State=\state]$ for $\data_i,\state\in\{0,1\}$. Then
\begin{equation*}
\begin{split}
&\mspace{23mu}U_i(p_1,p_0,q_1,q_0)\\
&=\expect_{(\sigma_i',\bm{\sigma}_{-i})}[R_i(\Databf)-\cost(\zeta(\sigma_i'))]\\
&=K_1p_1+K_0p_0+L_1q_1+L_0q_0-\cost(\zeta(p_1,p_0,q_1,q_0)),
\end{split}
\end{equation*}
with
\begin{align*}
K_1&=\{\overline{\paymecgen}_i(1;1)P_{\State}(1)\quality+\overline{\paymecgen}_i(1;0)P_{\State}(0)(1-\quality)\},\\
K_0&=\{\overline{\paymecgen}_i(1;1)P_{\State}(1)(1-\quality)+\overline{\paymecgen}_i(1;0)P_{\State}(0)\quality\},\\
L_1&=\{\overline{\paymecgen}_i(0;1)P_{\State}(1)\quality+\overline{\paymecgen}_i(0;0)P_{\State}(0)(1-\quality)\},\\
L_0&=\{\overline{\paymecgen}_i(0;1)P_{\State}(1)(1-\quality)+\overline{\paymecgen}_i(0;0)P_{\State}(0)\quality\}.
\end{align*}

In the designed mechanism $\paymecgenbf$, the payment to \player{}~$i$ only depends on $\Data_i$ and $M_{-i}$. Thus we write $\paymecgen_i(\Data_i;M_{-i})=\paymecgen_i(\Databf)$. Then the value of $\overline{\paymecgen}_i(\data_i;\state)$ is calculated as follows:
\begin{align*}
\overline{\paymecgen}_i(1;1)&=\expect_{(\sigma_i',\bm{\sigma}_{-i})}[\paymecgen_i(\Databf)\mid \Data_i=1,\State=1]\\
&=\beta^{(N)}\paymecgen_i(1;1)+(1-\beta^{(N)})\paymecgen_i(1;0),\\
\overline{\paymecgen}_i(1;0)&=\expect_{(\sigma_i',\bm{\sigma}_{-i})}[\paymecgen_i(\Databf)\mid \Data_i=1,\State=0]\\
&=(\gamma^{(N)}-\beta^{(N)})\paymecgen_i(1;1)+(1-(\gamma^{(N)}-\beta^{(N)}))\paymecgen_i(1;0),\\
\overline{\paymecgen}_i(0;1)&=\expect_{(\sigma_i',\bm{\sigma}_{-i})}[\paymecgen_i(\Databf)\mid \Data_i=0,\State=1]\\
&=(1-\beta^{(N)})\paymecgen_i(0;0)+\beta^{(N)}\paymecgen_i(0;1),\\
\overline{\paymecgen}_i(0;0)&=\expect_{(\sigma_i',\bm{\sigma}_{-i})}[\paymecgen_i(\Databf)\mid \Data_i=0,\State=0]\\
&=(1-(\gamma^{(N)}-\beta^{(N)}))\paymecgen_i(0;0)+(\gamma^{(N)}-\beta^{(N)})\paymecgen_i(0;1),
\end{align*}
and it can be verified that $K_1$, $K_0$, $L_1$ and $L_0$ are all positive. Therefore, by the proof of Lemma~\ref{lem:equi-char}, the possibility for the best response to be $p_1=p_0,q_1=q_0,0<p_1+q_1<1$ can be eliminated and the best response strategy must be in one of the following three forms:
\begin{gather}
p_1=p_0=q_1=q_0=0,\label{eq:not-participate}\\
p_1=p_0,\quad q_1=q_0,\quad p_1+q_1=1,\label{eq:non-informative}\\
p_1=q_0,\quad p_0=q_1,\quad p_1+q_1=1.\label{eq:symmetric}
\end{gather}
The strategy in \eqref{eq:not-participate} is to always not participate, which yields an utility of zero. For strategies in the form of \eqref{eq:non-informative} or \eqref{eq:symmetric}, we can write the expected utility as a function of $p_1$ and $p_0$ as follows:
\begin{equation*}
\overline{U}_i(p_1,p_0)=\overline{K}_1p_1+\overline{K}_0p_0+\overline{K}-\cost(\zeta(p_1,p_0)),
\end{equation*}
where $\overline{K}_1=K_1-L_1$, $\overline{K}_0=K_0-L_0$, $\overline{K}=L_1+L_0$, and with a little abuse of notation, $\zeta(p_1,p_0)=\max\biggl\{\biggl|\ln\frac{p_1}{p_0}\biggr|,\biggl|\ln\frac{1-p_1}{1-p_0}\biggr|\biggr\}$. Inserting the value of $\paymecgen_i(\Data_i;M_{-i})$ gives
\begin{equation*}
\overline{K}_1=\frac{\cost'(\epsilon)(e^{\epsilon}+1)^2}{2e^{\epsilon}},\quad\overline{K}_0=-\frac{\cost'(\epsilon)(e^{\epsilon}+1)^2}{2e^{\epsilon}}.
\end{equation*}
Then a strategy in the form of \eqref{eq:non-informative} yields an utility of $\overline{K}>0$. A strategy in the form of \eqref{eq:symmetric} can be written as
\begin{equation*}
p_1=q_0=\frac{e^{\xi}}{e^{\xi}+1},\quad p_0=q_1=\frac{1}{e^{\xi}+1}.
\end{equation*}
Then the expected utility can be further written as a function $f$ of $\xi$ as follows:
\begin{equation*}
f(\xi)=\overline{K}_1\frac{e^{\xi}}{e^{\xi}+1}+\overline{K}_0\frac{1}{e^{\xi}+1}-\cost(|\xi|)+\overline{K}.
\end{equation*}
Therefore, to prove that the $\epsilon$-strategy is a best response of \player{}~$i$, it suffices to prove that $\epsilon$ maximizes $f(\xi)$ and $f(\epsilon)\ge\overline{K}$. For any $\xi<0$, it is easy to see that
\begin{equation*}
\overline{K}_1\frac{e^{\xi}}{e^{\xi}+1}+\overline{K}_0\frac{1}{e^{\xi}+1}<0<\overline{K}_1\frac{e^{-\xi}}{e^{-\xi}+1}+\overline{K}_0\frac{1}{e^{-\xi}+1}.
\end{equation*}
Thus $f(\xi)$ achieves its maximum value at some $\xi\ge 0$. For any $\xi\ge 0$,
\begin{align*}
f'(\xi)&=(\overline{K}_1-\overline{K}_0)\frac{e^{\xi}}{(e^{\xi}+1)^2}-\cost'(\xi),\\
f''(\xi)&=-(\overline{K}_1-\overline{K}_0)\frac{e^{\xi}(e^{\xi}-1)}{(e^{\xi}+1)^3}-\cost''(\xi)\le 0,
\end{align*}
where the second inequality is due to the convexity of the cost function $\cost$. Therefore, $f$ is concave. Since $f'(\epsilon)=0$, $\epsilon$ maximizes $f(\xi)$. The optimal value is
\begin{align*}
f(\epsilon)&=\cost'(\epsilon)\frac{e^{\epsilon}-e^{-\epsilon}}{2}-\cost(\epsilon)+\overline{K}.
\end{align*}
By the convexity of $\cost$, $\cost(\epsilon)\le\cost'(\epsilon)\epsilon\le\cost'(\epsilon)\frac{e^{\epsilon}-e^{-\epsilon}}{2}$. Thus $f(\epsilon)\ge\overline{K}$, which completes the proof for the $\epsilon$-strategy to be a best response of \player{}~$i$.

Next we calculate the expected payment to \player{}~$i$ at $\bm{\sigma}$, which can be written as
\begin{equation*}
\expect_{\bm{\sigma}}[\paymecgen_i(\Databf)]=-(\overline{K}_1-\overline{K}_0)\frac{1}{e^{\epsilon}+1}+\overline{K}_1+\overline{K}.
\end{equation*}
By definitions,
\begin{align*}
&\mspace{23mu}\overline{K}_1+\overline{K}\\
&=\frac{g'(\epsilon)(e^{\epsilon}+1)^2}{2e^{\epsilon}}\frac{1}{(2\beta^{(N)}-\gamma^{(N)})(2\quality-1)}\\
&\mspace{23mu}\cdot\biggl(2\bigl(\beta^{(N)}\bigr)^2+(4\quality-2-2\gamma^{(N)})\beta^{(N)}\\
&\mspace{48mu}+2(1-\quality)\gamma^{(N)}+\beta^{(N)}(1-\beta^{(N)})\frac{P_{\State}(1)}{P_{\State}(0)}\\
&\mspace{48mu}+(\gamma^{(N)}-\beta^{(N)})(1-(\gamma^{(N)}-\beta^{(N)}))\frac{P_{\State}(0)}{P_{\State}(1)}\biggr)\\
&=: \frac{g'(\epsilon)(e^{\epsilon}+1)^2}{2e^{\epsilon}}h(\beta^{(N)}).
\end{align*}
Then
\begin{align*}
\expect_{\bm{\sigma}}[\paymecgen_i(\Databf)]&=\frac{\cost'(\epsilon)(e^{\epsilon}+1)}{e^{\epsilon}}\biggl(\frac{1}{2}h(\beta^{(N)})(e^{\epsilon}+1)-1\biggr)\\
&=\Value_{\mathrm{LB}}(\epsilon)+\frac{\cost'(\epsilon)(e^{\epsilon}+1)^2}{2e^{\epsilon}}\biggl(h(\beta^{(N)})-\frac{2\quality}{2\quality-1}\biggr).
\end{align*}

To derive an upper bound on the expected payment, we first analyze the function $h$. Rearranging terms gives
\begin{align*}
h(\beta^{(N)})&=\frac{1}{2\quality-1}\frac{1}{2\beta^{(N)}-\gamma^{(N)}}\\
&\mspace{23mu}\cdot\Biggl((2-t)\bigl(\beta^{(N)}\bigr)^2+\biggl(4\quality-2-2\gamma^{(N)}+\frac{P_{\State}(1)}{P_{\State}(0)}\\
&\mspace{54mu}+(2\gamma^{(N)}-1)\frac{P_{\State}(0)}{P_{\State}(1)}\biggr)\beta^{(N)}\\
&\mspace{54mu}+2(1-\quality)\gamma^{(N)}+\gamma^{(N)}(1-\gamma^{(N)})\frac{P_{\State}(0)}{P_{\State}(1)}\Biggr),
\end{align*}
where $t=\frac{(P_{\State}(1))^2+(P_{\State}(0))^2}{P_{\State}(1)P_{\State}(0)}\ge 2$. Taking derivative yields
\begin{align*}
h'(\beta^{(N)})&=\frac{1}{2\quality-1}\frac{1}{(2\beta^{(N)}-\gamma^{(N)})^2}\\
&\mspace{23mu}\cdot\biggl(2(2-t)\Bigl(\beta^{(N)}-\frac{\gamma^{(N)}}{2}\Bigr)^2-\bigl(\gamma^{(N)}\bigr)^2\\
&\mspace{50mu}-\frac{\gamma^{(N)}t}{2}(2-\gamma^{(N)})-2\gamma^{(N)}(1-\gamma^{(N)})\biggr).
\end{align*}
Therefore, $h'(\beta^{(N)})\le 0$ and $h$ is a non-increasing function.

Next we derive a lower bound on $\beta^{(N)}$. Let $Y_1,Y_2,\dots,Y_{N-1}$ be i.i.d.\ Bernoulli random variables with parameter $\alpha$. Then by the definition of $\beta^{(N)}$:
\begin{align*}
\beta^{(N)}&=\Pr\Biggl(\sum_{l=1}^{N-1}Y_l\ge\biggl\lfloor\frac{N-1}{2}\biggr\rfloor+1\Biggr)\\
&=\gamma^{(N)}-\Pr\Biggl(\sum_{l=1}^{N-1}(1-Y_l)\ge N-1-\biggl\lceil\frac{N-1}{2}\biggr\rceil+1\Biggr)\\
&\ge\gamma^{(N)}-\Pr\Biggl(\sum_{l=1}^{N-1}(1-Y_l)\ge\frac{N-1}{2}\Biggr)\\
&\ge \gamma^{(N)}-e^{-(N-1)d},
\end{align*}
where $d=\frac{1}{2}\ln\frac{1}{4\alpha(1-\alpha)}>0$ is the parameter defined in \eqref{eq:d} and the last inequality follow from the Chernoff bound \cite{SriYin_14}.

By the monotonicity of $h$,
\begin{align*}
&\mspace{25mu}h(\beta^{(N)})-\frac{2\quality}{2\quality-1}\\
&\le h\Bigl(\gamma^{(N)}-e^{-(N-1)d}\Bigr)-\frac{2\quality}{2\quality-1}\\
&=\frac{1}{2\quality-1}\frac{1}{\gamma^{(N)}-2e^{-(N-1)d}}\\
&\mspace{23mu}\cdot\Biggl((2-t)e^{-2(N-1)d}+\biggl(2(1-\gamma^{(N)})+2\gamma^{(N)}t\\
&\mspace{54mu}-\frac{P_{\State}(1)}{P_{\State}(0)}-(2\gamma^{(N)}-1)\frac{P_{\State}(0)}{P_{\State}(1)}\biggr)e^{-(N-1)d}\\
&\mspace{54mu}+\gamma^{(N)}\frac{P_{\State}(1)}{P_{\State}(0)}+\bigl(\gamma^{(N)}\bigr)^2\frac{P_{\State}(0)}{P_{\State}(1)}-\bigl(\gamma^{(N)}\bigr)^2t\Biggr)\\
&\le\frac{1}{2\quality-1}\frac{1}{\gamma^{(N)}-2e^{-(N-1)d}}\\
&\mspace{23mu}\cdot\biggl((2-t)e^{-2(N-1)d}+(2(1-\gamma^{(N)})+t)e^{-(N-1)d}\\
&\mspace{54mu}+\gamma^{(N)}(1-\gamma^{(N)})\frac{P_{\State}(1)}{P_{\State}(0)}\biggr).
\end{align*}

Notice that
\begin{equation*}
1-\gamma^{(N)}=
\begin{cases}
\displaystyle
\binom{N-1}{\frac{N-1}{2}}\alpha^{\frac{N-1}{2}}(1-\alpha)^{\frac{N-1}{2}}& \text{if $N-1$ is even,}\\
0 & \text{if $N-1$ is odd.}
\end{cases}
\end{equation*}
Then when $N-1$ is odd, $\gamma^{(N)}=1$, and when $N-1$ is even,
\begin{align*}
1-\gamma^{(N)}&=\binom{N-1}{\frac{N-1}{2}}\alpha^{\frac{N-1}{2}}(1-\alpha)^{\frac{N-1}{2}}\\
&=e^{-(N-1)d}\cdot\binom{N-1}{\frac{N-1}{2}}2^{-(N-1)},
\end{align*}
where $\lim_{N\rightarrow\infty}\binom{N-1}{\frac{N-1}{2}}2^{-(N-1)}=0$. Thus $1-\gamma^{(N)}=O(e^{-Nd})$ as $N\rightarrow\infty$.

Therefore,
\begin{align*}
&\mspace{23mu}\expect_{\bm{\sigma}}[\paymecgen_i(\Databf)]\\
&\le\Value_{\mathrm{LB}}(\epsilon)+\frac{\cost'(\epsilon)(e^{\epsilon}+1)^2}{2e^{\epsilon}}\biggl(h\Bigl(\gamma^{(N)}-e^{-(N-1)d}\Bigr)-\frac{2\quality}{2\quality-1}\biggr)\\
&\le\Value_{\mathrm{LB}}(\epsilon)+\frac{\cost'(\epsilon)(e^{\epsilon}+1)^2}{2e^{\epsilon}}\frac{1}{2\quality-1}\frac{1}{\gamma^{(N)}-2e^{-(N-1)d}}\\
&\mspace{23mu}\cdot\Bigl((2-t)e^{-2(N-1)d}+(2(1-\gamma^{(N)})+t)e^{-(N-1)d}+O(e^{-Nd})\Bigr)\\
&=\Value_{\mathrm{LB}}(\epsilon)+O(e^{-Nd}),
\end{align*}
as $N\rightarrow\infty$, which completes the proof.
\end{proof}

\subsection{Extension to Heterogeneous Cost Functions}\label{sec:hete}
Our results on the value of privacy are also valid in the scenario where \players{}' privacy cost functions are heterogeneous and known. In this case, the value of $\epsilon$ units of privacy is still measured by the minimum payment of all nonnegative payment mechanisms under which an \player{}'s best response in a Nash equilibrium is to report the data with a privacy level of $\epsilon$. However, with heterogeneous cost functions, this value differs from \player{} to \player{}. Following similar notation, we let $\Value_i(\epsilon)$ denote the value of $\epsilon$ units of privacy to \player{}~$i$, and let $\cost_i$ denote the cost function of \player{}~$i$. Then the following lower and upper bounds, which are almost identical to those in Theorem~\ref{thm:lower} and \ref{thm:upper} except the heterogeneous cost function $\cost_i(\epsilon)$, hold
\begin{equation*}
\cost_i'(\epsilon)\frac{e^{\epsilon}+1}{e^{\epsilon}}\biggl(\frac{\quality}{2\quality-1}(e^{\epsilon}+1)-1\biggr) \le\Value_i(\epsilon)
\le \cost_i'(\epsilon)\frac{e^{\epsilon}+1}{e^{\epsilon}}\biggl(\frac{\quality}{2\quality-1}(e^{\epsilon}+1)-1\biggr)+O(e^{-Nd}).
\end{equation*}
The lower bound above can be derived directly from the proof of Theorem~\ref{thm:lower}, since the proof does not depend on whether the cost functions are homogeneous or not. The upper bound above is given by a payment mechanism that works similar to $\paymecbf$, with the $\cost'$ in $\paymec_i$ replaced by $\cost'_i$. In this mechanism, the strategy profile $\bm{\sigma}^{(\epsilon)}$ is still a Nash equilibrium, and the expected payment to \player{}~$i$ at this equilibrium can still be upper bounded as in Theorem~\ref{thm:upper} but again with $\cost'$ replaced by $\cost'_i$.

\section{Payment vs.\ Accuracy}
\begin{sloppypar}
In this section, we apply the fundamental bounds on the value of privacy to the payment--accuracy problem, where the \principal{} aims to minimize the total payment while achieving an accuracy target in learning the state. The solution of this problem can be used to guide the design of review systems. For example, to evaluate the underlying value of a new product, a review system can utilize the results in this section to design a payment mechanism for eliciting informative feedback from testers.
\end{sloppypar}

\subsection{Payment--Accuracy Problem}
The \principal{} learns the state $\State$ from the reported data $\Data_1,\Data_2,\dots,\Data_N$, which is collected through some payment mechanism, by performing hypothesis testing between the following two hypotheses:
\begin{align*}
H_0\colon & \State=0,\\
H_1\colon & \State=1.
\end{align*}
The conditional distributions of the reported data given the hypotheses are specified by the strategy profile in a Nash equilibrium of the payment mechanism. According to Lemma~\ref{lem:equi-char}, we can write an equilibrium strategy profile in the form of $(\sigma_i^{(\epsilon_i)})=(\sigma_1^{(\epsilon_1)},\sigma_2^{(\epsilon_2)},\dots,\sigma_N^{(\epsilon_N)})$ with $\epsilon_i\in\mathbb{R}\setminus\{0\}\cup\{\independent\}$, where recall that $\sigma_i^{(\epsilon_i)}$ is the $\epsilon_i$-strategy. For ease of notation, a non-informative strategy is also called an $\epsilon$-strategy but with $\epsilon=\independent$. Let $\mathcal{\paymecgenbf}(\epsilon_1,\epsilon_2,\dots,\epsilon_N)$ denote the set of nonnegative payment mechanisms in which $(\sigma_i^{(\epsilon_i)})$ is a Nash equilibrium.

We consider an information-theoretic approach based on the Chernoff information \cite{CovTho_06} to measure the accuracy that can be achieved in hypothesis testing. For each \player{}~$i$, let $D(\epsilon_i)$ denote the Chernoff information between the conditional distributions of $\Data_i$ given $\State=1$ and $\State=0$. The larger $D(\epsilon_i)$ is, the more possible that the two hypotheses can be distinguished. In later discussions we will see that the Chernoff information is closely related to the best achievable probability of error.

The \principal{} aims to minimize the total payment while achieving an accuracy target. The design choices include the number of \players{} $N$, the parameters $\epsilon_1,\epsilon_2,\dots,\epsilon_N$, and the payment mechanism $\paymecgenbf$ in which the strategy profile $(\sigma_i^{(\epsilon_i)})$ is a Nash equilibrium. Then we formulate the mechanism design problem as the following optimization problem, which we call the \emph{payment--accuracy problem}:
\begin{align*}
\underset{\substack{N\in\naturalnumber,\,\epsilon_i\in\mathbb{R}\setminus\{0\}\cup\{\independent\},\forall i\\\paymecgenbf\in\mathcal{R}(\epsilon_1,\epsilon_2,\dots,\epsilon_N)}
}{\text{min}}\mspace{12mu}&\sum_{i=1}^{N}\expect_{(\sigma_i^{(\epsilon_i)})}[\paymecgen_i(\Databf)]\\
\text{subject to}\mspace{30mu}& e^{-\sum_{i=1}^{N}D(\epsilon_i)}\le \err,
\end{align*}
where the accuracy target is represented by $\err$, which is related to the maximum allowable error. We focus on the range $\err\in(0,1)$ for nontriviality. Let $F(\err)$ denote the optimal payment in this problem, i.e., the infimum of the total payment while satisfying the accuracy target $\err$.

\subsection{Bounds on the Payment--Accuracy Problem}
We present bounds on $F(\err)$ in Theorem~\ref{thm:payment--accuracy} below. For convenience, we define
\begin{equation}\label{eq:tilda}
\widetilde{\epsilon}=\inf\biggl\{\arg\max\biggl\{\frac{D(\epsilon)}{\Value_{\mathrm{LB}}(\epsilon)}\colon \epsilon>0\biggr\}\biggr\},\mspace{9mu}
\widetilde{N}=\biggl\lceil\frac{\ln(1/\err)}{D(\widetilde{\epsilon})}\biggr\rceil,
\end{equation}
where recall that $\Value_{\mathrm{LB}}(\epsilon)$ is the lower bound in Theorem~\ref{thm:lower}.
\begin{sloppypar}
\begin{theorem}\label{thm:payment--accuracy}
The optimal payment $F(\err)$ in the payment--accuracy problem for a given accuracy target $\err\in(0,1)$ is bounded as: $(\widetilde{N}-1)\Value_{\mathrm{LB}}(\widetilde{\epsilon})\le F(\err)\le \widetilde{N}\Value_{\mathrm{LB}}(\widetilde{\epsilon})+O(\err\ln(1/\err))$, where the $O(\cdot)$ is for $\err\rightarrow 0$.
\end{theorem}
\end{sloppypar}

The upper bound in Theorem~\ref{thm:payment--accuracy} is given by the designed mechanism $\paymecbf$ with parameters chosen as $\epsilon=\widetilde{\epsilon}$ and $N=\widetilde{N}$. Note that $\widetilde{\epsilon}$ can be proved to have a well-defined finite value independent of $\err$. By the lower and upper bounds on the value of privacy, the payment to each \player{} in $\paymecgenbf^{(\widetilde{N},\widetilde{\epsilon})}$ is roughly equal to the lower bound $\Value_{\mathrm{LB}}(\widetilde{\epsilon})$. Then Theorem~\ref{thm:payment--accuracy} indicates that the total payment of the designed mechanism $\paymecgenbf^{(\widetilde{N},\widetilde{\epsilon})}$ is at most one \player{}'s payment away from the minimum, with the diminishing term $O(\tau\ln(1/\tau))$ omitted. Figure~\ref{fig:paymentAccuracy} shows an illustration of the lower and upper bounds.
\begin{figure}
\centering
\includegraphics[scale=0.35]{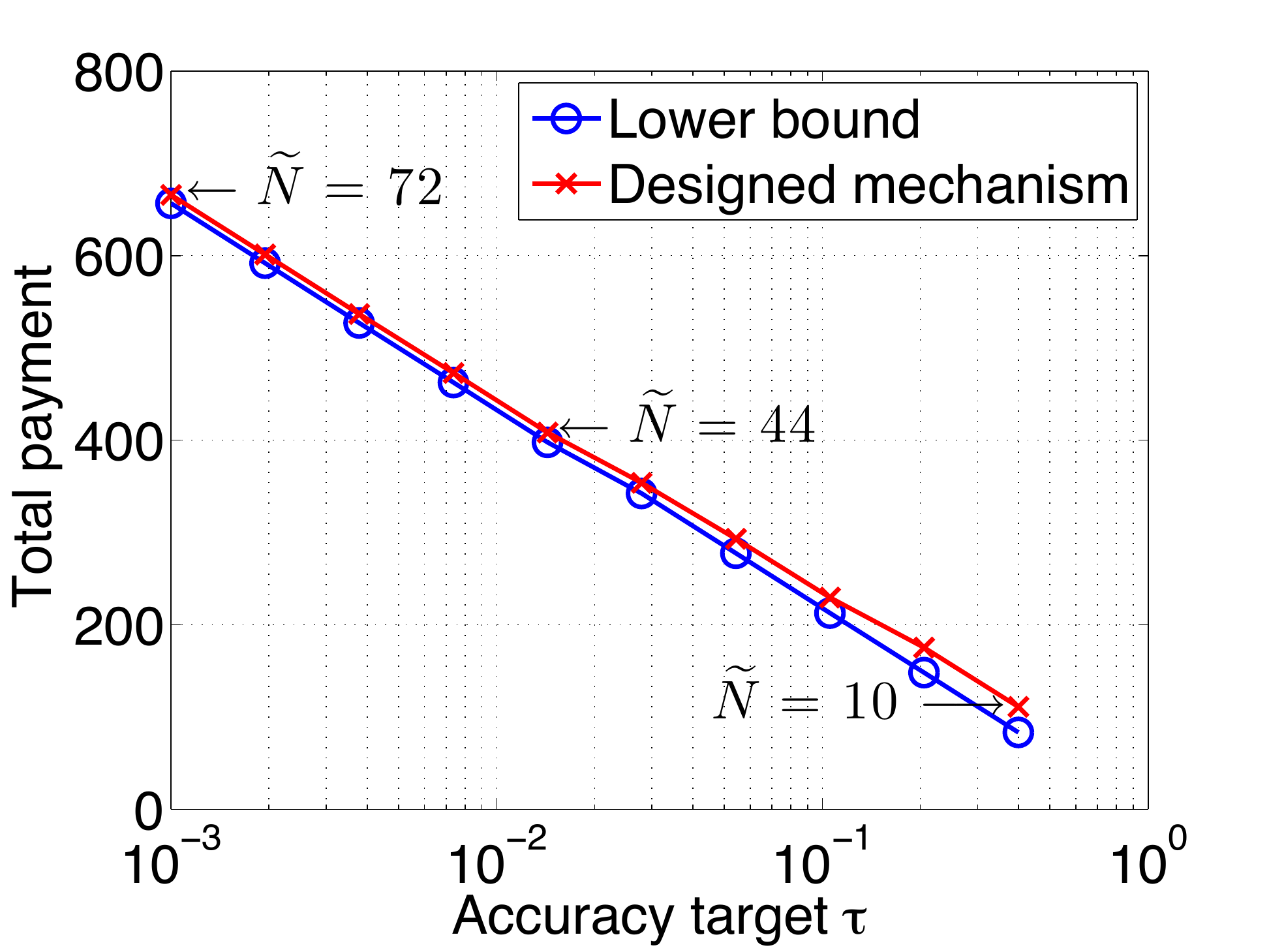}
\caption{Illustration of the lower and upper bounds in Theorem~\ref{thm:payment--accuracy} on the minimum total payment for achieving an accuracy target $\err$, where the upper bound is given by the designed mechanism $\paymecgenbf^{(\widetilde{N},\widetilde{\epsilon})}$. In this example, the prior \pmf\ of the state is $P_{\State}(1)=0.7$, $P_{\State}(0)=0.3$. The quality of signals is $\quality=0.8$. The cost function is $\cost(\epsilon)=\epsilon$. The range of $\err$ shown in the figure is $0.001$--$0.4$.}
\label{fig:paymentAccuracy}
\end{figure}

Theorem~\ref{thm:payment--accuracy} is proved by Lemma~\ref{lem:lower-P2} and Lemma~\ref{lem:upper-mec} below, where the lower bound is given by the lower bound on the value of privacy, and the upper bound is given by $\paymecgenbf^{(\widetilde{N},\widetilde{\epsilon})}$.

\subsubsection{Lower Bound}
First, notice that it suffices to limit the choice of each $\epsilon_i$ to $(0,+\infty)$ in the payment--accuracy problem, since when $\epsilon_i=\independent$, $D(\epsilon_i)=0$, and when $\epsilon_i<0$, $D(\epsilon_i)=D(|\epsilon_i|)$ and there exists another nonnegative payment mechanism with the same payment property and a Nash equilibrium at $(\sigma_i^{(|\epsilon_i|)})$ by Lemma~\ref{lem:positive-eps}.

Now we use the lower bound on the value of privacy to prove the lower bound on $F(\err)$. By Theorem~\ref{thm:lower},
\begin{equation*}
\inf_{\paymecgenbf\in\mathcal{R}{(\epsilon_1,\epsilon_2,\dots,\epsilon_N)}}\sum_{i=1}^{N}\expect_{(\sigma_i^{(\epsilon_i)})}[\paymecgen_i(\Databf)]\ge\sum_{i=1}^N\Value_{\mathrm{LB}}(\epsilon_i).
\end{equation*}
Therefore, the optimal payment $F(\err)$ is lower bounded by the optimal value of the following optimization problem \eqref{eq:payment-accuracy-lower}:
\begin{equation}\tag{P1}\label{eq:payment-accuracy-lower}
\begin{split}
\underset{N\in\mathbb{N},\,\epsilon_i\in(0,+\infty),\forall i}{\text{min}}\mspace{12mu}& \sum_{i=1}^N\Value_{\mathrm{LB}}(\epsilon_i)\\
\text{subject to}\mspace{30mu}& e^{-\sum_{i=1}^N D(\epsilon_i)}\le \err.\nonumber
\end{split}
\end{equation}
\begin{sloppypar}
\begin{lemma}\label{lem:lower-P2}
Any feasible solution $(N,\epsilon_1,\epsilon_2,\dots,\epsilon_N)$ of \eqref{eq:payment-accuracy-lower} satisfies
\begin{equation*}
\sum_{i=1}^{N}\Value_{\mathrm{LB}}(\epsilon_i)\ge(\widetilde{N}-1)\Value_{\mathrm{LB}}(\widetilde{\epsilon}),
\end{equation*}
where $\widetilde{\epsilon}$ and $\widetilde{N}$ are defined in \eqref{eq:tilda}.
\end{lemma}
Lemma~\ref{lem:lower-P2} states that the total expected payment of the data collector is at least $(\widetilde{N}-1)\Value_{\mathrm{LB}}(\widetilde{\epsilon})$. Note that the value given by the genie-aided payment mechanism $\widehat{\paymecgenbf}^{(\widetilde{\epsilon})}$ for $\widetilde{N}$ \players{} is $\widetilde{N}\Value_{\mathrm{LB}}(\widetilde{\epsilon})$, which is at most one $\Value_{\mathrm{LB}}(\widetilde{\epsilon})$ away from the optimal value of \eqref{eq:payment-accuracy-lower}. We can think of $\Value_{\mathrm{LB}}(\epsilon)$ as the price for $\epsilon$ units of privacy and $D(\epsilon)$ as the quality that the \principal{} gets from $\epsilon$ units of privacy due to its contribution to the accuracy. Then the intuition for $(\widetilde{N},\widetilde{\epsilon},\dots,\widetilde{\epsilon})$ to be a near-optimal choice is that the privacy level $\widetilde{\epsilon}$ gives the best quality\slash price ratio and $\widetilde{N}$ is the fewest number of \players{} to meet the accuracy target. The proof of Lemma~\ref{lem:lower-P2} is presented is Appendix~\ref{app:proof-lower-P2}. With this lemma, the lower bound on $F(\err)$ in Theorem~\ref{thm:payment--accuracy} is straightforward.
\end{sloppypar}

\subsubsection{Upper Bound}
\begin{lemma}\label{lem:upper-mec}
Choose the parameters in the payment mechanism $\paymecbf$ defined in Section~\ref{sec:paymec} to be $\epsilon=\widetilde{\epsilon}$ and $N=\widetilde{N}$, where $\widetilde{\epsilon}$ and $\widetilde{N}$ are defined in \eqref{eq:tilda}. Then in the Nash equilibrium $\bm{\sigma}^{(\widetilde{\epsilon})}$ of $\paymecgenbf^{(\widetilde{N},\widetilde{\epsilon})}$, the accuracy target $\err$ can be achieved, and the total expected payment is upper bounded as
\begin{equation*}
\expect_{\bm{\sigma}^{(\widetilde{\epsilon})}}\Biggl[\sum_{i=1}^{\widetilde{N}}\paymecgen^{(\widetilde{N},\widetilde{\epsilon})}_i(\Databf)\Biggr]\le\widetilde{N}\Value_{\mathrm{LB}}(\widetilde{\epsilon})+O(\err\ln(1/\err)).
\end{equation*}
\end{lemma}
This lemma follows from Theorem~\ref{thm:upper} and we omit the proof here. Since the payment mechanism $\paymecbf$ together with $\epsilon=\widetilde{\epsilon}$ and $N=\widetilde{N}$ is a feasible solution of the payment--accuracy problem, the upper bound in this lemma gives the upper bound on $F(\err)$ in Theorem~\ref{thm:payment--accuracy}.

\subsection{Discussions on the Accuracy Metric}
When we study the relation between payment and accuracy, the accuracy can also be measured by the best achievable probability of error, defined as
\begin{equation*}
p_e=\inf_{\psi}\Pr_{(\sigma_i^{(\epsilon_i)})}(\psi(\Databf)\neq \State),
\end{equation*}
where $\psi(\databf)$ is a decision function, with $\psi(\databf)=0$ implying that $H_0$ is accepted and $\psi(\databf)=1$ implying that $H_1$ is accepted. However, $p_e$ is difficult to deal with analytically since its exact form in terms of $\epsilon_1,\epsilon_2,\dots,\epsilon_N$ is intractable.

We measure the accuracy based on the Chernoff information, which is an information-theoretic metric closely related to $p_e$. It can be proved by the Bhattacharyya bound \cite{Kai_67} that at the strategy profile $(\sigma_i^{(\epsilon_i)})$,
\begin{equation}\label{eq:pe-upper}
p_e\le e^{-\sum_{i=1}^N D(\epsilon_i)}.
\end{equation}
Therefore, if we want to guarantee that $p_e\le p_e^{\text{max}}$ for some maximum allowable probability of error $p_e^{\text{max}}$, we can choose $\err=p_e^{\text{max}}$ in the payment--accuracy problem. In fact, the metric based on the Chernoff information is very close to the metric $p_e$, since the upper bound \eqref{eq:pe-upper} is tight in exponent when all the $\epsilon_i$ are the same, i.e., when the reported data is i.i.d.\ given the hypothesis.

\section{Conclusions}
In this paper, we studied ``the value of privacy'' under a game-theoretic model, where a \principal{} pays strategic \players{} to buy their private data for a learning purpose. The \players{} do not consider the \principal{} to be trustworthy, and thus experience a cost of privacy loss during data reporting. The value of $\epsilon$ units of privacy is measured by the minimum payment of all nonnegative payment mechanisms under which an \player{}'s best response in a Nash equilibrium is to report the data with a privacy level of $\epsilon$. We derived asymptotically tight lower and upper bounds on the value of privacy as the number of \players{} becomes large, where the upper bound was given by a designed payment mechanism $\paymecbf$. We further applied these fundamental limits to find the minimum total payment for the \principal{} to achieve certain learning accuracy target, and derived lower and upper bounds on the minimum payment. The total payment of the designed mechanism $\paymecbf$ with properly chosen parameters is at most one \player{}'s payment away from the minimum. It would be of great interest to study the value of privacy under $(\epsilon,\delta)$-differential privacy, and to extend our results to more general models of private and reported data, e.g., models with larger alphabets for the state, the signals and the reported data.

\section{Acknowledgement}
This work was supported in part by the NSF under Grant ECCS-1255425.

\bibliographystyle{acm}
\bibliography{U:/bib/inlab-refs.bib}

\begin{thebibliography}{10}

\bibitem{AceDahLob_11}
{\sc Acemoglu, D., Dahleh, M.~A., Lobel, I., and Ozdaglar, A.}
\newblock Bayesian learning in social networks.
\newblock {\em Review of Econ. Stud. 78}, 4 (Oct. 2011), 1201--1236.

\bibitem{BasSmi_15}
{\sc Bassily, R., and Smith, A.}
\newblock Local, private, efficient protocols for succinct histograms.
\newblock In {\em Proc. Ann. ACM Symp. Theory of Computing (STOC)\/} (Portland,
  OR, 2015), pp.~127--135.

\bibitem{CheChoKas_13}
{\sc Chen, Y., Chong, S., Kash, I.~A., Moran, T., and Vadhan, S.}
\newblock Truthful mechanisms for agents that value privacy.
\newblock In {\em Proc. ACM Conf. Electronic Commerce (EC)\/} (Philadelphia,
  PA, 2013), pp.~215--232.

\bibitem{CheSheVad_14}
{\sc Chen, Y., Sheffet, O., and Vadhan, S.}
\newblock Privacy games.
\newblock In {\em Int. Conf. Web and Internet Economics (WINE)\/} (2014),
  vol.~8877, pp.~371--385.

\bibitem{CovTho_06}
{\sc Cover, T.~M., and Thomas, J.~A.}
\newblock {\em Elements of Information Theory}, 2nd~ed.
\newblock John Wiley \& Sons, Hoboken, NJ, 2006.

\bibitem{DucJorWai_13}
{\sc Duchi, J.~C., Jordan, M.~I., and Wainwright, M.~J.}
\newblock Local privacy and minimax bounds: Sharp rates for probability
  estimation.
\newblock In {\em Advances Neural Information Processing Systems (NIPS)\/}
  (Lake Tahoe, NV, Dec. 2013), pp.~1529--1537.

\bibitem{Dwo_06}
{\sc Dwork, C.}
\newblock Differential privacy.
\newblock In {\em Proc. Int. Conf. Automata, Languages and Programming
  (ICALP)\/} (Venice, Italy, 2006), pp.~1--12.

\bibitem{DwoMcSNis_06}
{\sc Dwork, C., McSherry, F., Nissim, K., and Smith, A.}
\newblock Calibrating noise to sensitivity in private data analysis.
\newblock In {\em Proc. Conf. Theory of Cryptography (TCC)\/} (New York, NY,
  2006), pp.~265--284.

\bibitem{DwoRot_14}
{\sc Dwork, C., and Roth, A.}
\newblock The algorithmic foundations of differential privacy.
\newblock {\em Found. Trends Theor. Comput. Sci. 9}, 3--4 (Aug. 2014),
  211--407.

\bibitem{ErlPihKor_14}
{\sc Erlingsson, {\'{U}}., Pihur, V., and Korolova, A.}
\newblock {RAPPOR}: Randomized aggregatable privacy-preserving ordinal
  response.
\newblock In {\em Proc. ACM SIGSAC Conf. Computer and Communication Security
  (CCS)\/} (Scottsdale, AZ, 2014), pp.~1054--1067.

\bibitem{FanPihErl_15}
{\sc Fanti, G.~C., Pihur, V., and Erlingsson, {\'{U}}.}
\newblock Building a {RAPPOR} with the unknown: Privacy-preserving learning of
  associations and data dictionaries.
\newblock {\em arXiv:1503.01214 [cs.CR]\/} (2015).

\bibitem{FleLyu_12}
{\sc Fleischer, L.~K., and Lyu, Y.}
\newblock Approximately optimal auctions for selling privacy when costs are
  correlated with data.
\newblock In {\em Proc. ACM Conf. Electronic Commerce (EC)\/} (Valencia, Spain,
  2012), pp.~568--585.

\bibitem{GhoLig_13}
{\sc Ghosh, A., and Ligett, K.}
\newblock Privacy and coordination: Computing on databases with endogenous
  participation.
\newblock In {\em Proc. ACM Conf. Electronic Commerce (EC)\/} (Philadelphia,
  PA, 2013), pp.~543--560.

\bibitem{GhoLigRot_14}
{\sc Ghosh, A., Ligett, K., Roth, A., and Schoenebeck, G.}
\newblock Buying private data without verification.
\newblock In {\em Proc. ACM Conf. Economics and Computation (EC)\/} (Palo Alto,
  CA, 2014), pp.~931--948.

\bibitem{GhoRot_11}
{\sc Ghosh, A., and Roth, A.}
\newblock Selling privacy at auction.
\newblock In {\em Proc. ACM Conf. Electronic Commerce (EC)\/} (San Jose, CA,
  2011), pp.~199--208.

\bibitem{HsuKhaRo_12}
{\sc Hsu, J., Khanna, S., and Roth, A.}
\newblock Distributed private heavy hitters.
\newblock In {\em Proc. Int. Conf. Automata, Languages and Programming
  (ICALP)\/} (Warwick, UK, 2012), pp.~461--472.

\bibitem{Kai_67}
{\sc Kailath, T.}
\newblock The divergence and {Bhattacharyya} distance measures in signal
  selection.
\newblock {\em IEEE Trans. Commun. Technol. 15}, 1 (Feb. 1967), 52--60.

\bibitem{KaiOhVis_14}
{\sc Kairouz, P., Oh, S., and Viswanath, P.}
\newblock Extremal mechanisms for local differential privacy.
\newblock In {\em Advances Neural Information Processing Systems (NIPS)\/}
  (Montreal, Canada, Dec. 2014), pp.~2879--2887.

\bibitem{KasLeeNis_11}
{\sc Kasiviswanathan, S.~P., Lee, H.~K., Nissim, K., Raskhodnikova, S., and
  Smith, A.}
\newblock What can we learn privately?
\newblock {\em SIAM J. Comput. 40}, 3 (May 2011), 793--826.

\bibitem{Kro_14}
{\sc Kroft, S.}
\newblock The data brokers: selling your personal information.
\newblock {\em {CBS} News\/} (Mar. 2014).

\bibitem{LeSubBer_14}
{\sc Le, T.~N., Subramanian, V.~G., and Berry, R.~A.}
\newblock The value of noise for informational cascades.
\newblock In {\em Proc. IEEE Int. Symp. Information Theory (ISIT)\/} (Honolulu,
  HI, 2014), pp.~1101--1105.

\bibitem{LigRot_12}
{\sc Ligett, K., and Roth, A.}
\newblock Take it or leave it: Running a survey when privacy comes at a cost.
\newblock In {\em Proc. Int. Workshop Internet and Network Economics (WINE)\/}
  (Liverpool, UK, 2012), pp.~378--391.

\bibitem{MilResZec_09}
{\sc Miller, N., Resnick, P., and Zeckhauser, R.}
\newblock Eliciting informative feedback: The peer-prediction method.
\newblock In {\em Computing with Social Trust}, Human--Computer Interaction
  Series. Springer London, 2009, pp.~185--212.

\bibitem{NisVadXia_14}
{\sc Nissim, K., Vadhan, S., and Xiao, D.}
\newblock Redrawing the boundaries on purchasing data from privacy-sensitive
  individuals.
\newblock In {\em Proc. Conf. Innovations in Theoretical Computer Science
  (ITCS)\/} (Princeton, NJ, 2014), pp.~411--422.

\bibitem{PaiRot_13}
{\sc Pai, M.~M., and Roth, A.}
\newblock Privacy and mechanism design.
\newblock {\em SIGecom Exch. 12}, 1 (June 2013), 8--29.

\bibitem{RotSch_12}
{\sc Roth, A., and Schoenebeck, G.}
\newblock Conducting truthful surveys, cheaply.
\newblock In {\em Proc. ACM Conf. Electronic Commerce (EC)\/} (Valencia, Spain,
  2012), pp.~826--843.

\bibitem{Sho_15}
{\sc Shokri, R.}
\newblock Privacy games: Optimal user-centric data obfuscation.
\newblock In {\em Proc. Privacy Enhancing Technologies (PETS)\/} (Philadelphia,
  PA, 2015), pp.~299--315.

\bibitem{SriYin_14}
{\sc Srikant, R., and Ying, L.}
\newblock {\em Communication Networks: An Optimization, Control and Stochastic
  Networks Perspective}.
\newblock Cambridge Univ. Press, New York, 2014.

\bibitem{WanYinZha_14}
{\sc Wang, W., Ying, L., and Zhang, J.}
\newblock On the relation between identifiability, differential privacy, and
  mutual-information privacy.
\newblock In {\em Proc. Ann. Allerton Conf. Communication, Control and
  Computing\/} (Monticello, IL, Sept. 2014), pp.~1086--1092.

\bibitem{WanYinZha_15}
{\sc Wang, W., Ying, L., and Zhang, J.}
\newblock A minimax distortion view of differentially private query release.
\newblock In {\em Proc. Asilomar Conf. Signals, Systems, and Computers\/}
  (Pacific Grove, CA, Nov. 2015).

\bibitem{War_65}
{\sc Warner, S.~L.}
\newblock Randomized response: A survey technique for eliminating evasive
  answer bias.
\newblock {\em J. Amer. Stat. Assoc. 60}, 309 (Mar. 1965), 63--69.

\bibitem{Xia_13}
{\sc Xiao, D.}
\newblock Is privacy compatible with truthfulness?
\newblock In {\em Proc. Conf. Innovations in Theoretical Computer Science
  (ITCS)\/} (Berkeley, CA, 2013), pp.~67--86.

\end{thebibliography}

\appendix
\numberwithin{equation}{section}

\def\thmNoLower{\ref{thm:lower}}
\def\lemNoEqui{\ref{lem:equi-char}}
\section{Proof of Lemma~\lemNoEqui}\label{app:lem:equi-char}
\begin{proof}
Consider any nonnegative payment mechanism $\paymecgenbf$ and a Nash equilibrium of it, denoted by $\bm{\sigma}$. For an \player{}~$i$, consider any strategy $\sigma_i'$ of \player{}~$i$ and let
\begin{equation*}
\begin{split}
p_1=\Pr_{\sigma_i'}(\Data_i=1\mid \Signal_i=1),\quad& q_1=\Pr_{\sigma_i'}(\Data_i=0\mid \Signal_i=1),\\
p_0=\Pr_{\sigma_i'}(\Data_i=1\mid \Signal_i=0),\quad& q_0=\Pr_{\sigma_i'}(\Data_i=0\mid \Signal_i=0).
\end{split}
\end{equation*}
When other \players{} follow $\bm{\sigma}_{-i}$, the expected utility of \player{}~$i$ at the strategy $\sigma_i'$ is a function of $(p_1,p_0,q_1,q_0)$, denoted by $U_i(p_1,p_0,q_1,q_0)$. We derive the form of this function below. The expected payment to \player{}~$i$ can be written as
\begin{align*}
&\mspace{23mu}\expect_{(\sigma_i',\bm{\sigma}_{-i})}[\paymecgen_i(\Databf)]\\
&=\sum_{\data_i,\signalbf}\biggl\{\Pr_{\sigma_i'}(\Data_i=\data_i,\Signal_i=\signal_i,\Signalbf_{-i}=\signalbf_{-i})\\
&\mspace{63mu}\cdot\expect_{(\sigma_i',\bm{\sigma}_{-i})}[\paymecgen_i(\Databf)\mid \Data_i=\data_i,\Signal_i=\signal_i,\Signalbf_{-i}=\signalbf_{-i}]\biggr\}\\
&=\sum_{\data_i,\signalbf}\biggl\{\Pr_{\sigma_i'}(\Data_i=\data_i\mid \Signal_i=\signal_i)\Pr(\Signal_i=\signal_i,\Signalbf_{-i}=\signalbf_{-i})\\
&\mspace{63mu}\cdot\expect_{(\sigma_i',\bm{\sigma}_{-i})}[\paymecgen_i(\Databf)\mid \Data_i=\data_i,\Signalbf_{-i}=\signalbf_{-i}]\biggr\},
\end{align*}
where we have used the fact that $\Data_i$ is independent from $\Signalbf_{-i}$ given $\Signal_i$, and $\Databf_{-i}$ is independent from $\Signal_i$ given $\Data_i$ and $\Signalbf_{-i}$. The term $\expect_{(\sigma_i',\bm{\sigma}_{-i})}[\paymecgen_i(\Databf)\mid \Data_i=\data_i,\Signalbf_{-i}=\signalbf_{-i}]$ does not depend on the strategy of \player{}~$i$ since
\begin{align*}
&\mspace{23mu}\expect_{(\sigma_i',\bm{\sigma}_{-i})}[\paymecgen_i(\Databf)\mid \Data_i=\data_i,\Signalbf_{-i}=\signalbf_{-i}]\\
&=\expect_{(\sigma_i',\bm{\sigma}_{-i})}[\paymecgen_i(\data_i,\Databf_{-i})\mid \Data_i=\data_i,\Signalbf_{-i}=\signalbf_{-i}]\\
&=\expect_{\bm{\sigma}_{-i}}[\paymecgen_i(\data_i,\Databf_{-i})\mid \Signalbf_{-i}=\signalbf_{-i}],
\end{align*}
where the last equality follow from the conditional independence between $\Data_i$ and $\Databf_{-i}$ given $\Signalbf_{-i}$. Then
\begin{align*}
&\mspace{23mu}\expect_{(\sigma_i',\bm{\sigma}_{-i})}[\paymecgen_i(\Databf)]\\
&=\sum_{\data_i,\signal_i}\biggl\{\Pr_{\sigma_i'}(\Data_i=\data_i\mid \Signal_i=\signal_i)\\
&\mspace{63mu}\cdot\sum_{\signalbf_{-i}}\Bigl(\Pr(\Signalbf=\signalbf)\expect_{\bm{\sigma}_{-i}}[\paymecgen_i(\data_i,\Databf_{-i})\mid \Signalbf_{-i}=\signalbf_{-i}]\Bigr)\biggr\}\\
&=K_1p_1+K_0p_0+L_1q_1+L_0q_0,
\end{align*}
where
\begin{align*}
K_{\signal_i}&=\sum_{\signalbf_{-i}}\Bigl(\Pr(\Signal_i=\signal_i,\Signalbf_{-i}=\signalbf_{-i})\\
&\mspace{54mu}\cdot\expect_{\bm{\sigma}_{-i}}[\paymecgen_i(1,\Databf_{-i})\mid \Signalbf_{-i}=\signalbf_{-i}]\Bigr),\mspace{6mu}\signal_i\in\{0,1\},
\end{align*}
are the expected payment received by \player{}~$i$ when she reports $1$, weighted by $\Pr(\Signal_i=1)$ and $\Pr(\Signal_i=0)$ when her private signal is $1$ and $0$, respectively, and
\begin{align*}
L_{\signal_i}&=\sum_{\signalbf_{-i}}\Bigl(\Pr(\Signal_i=\signal_i,\Signalbf_{-i}=\signalbf_{-i})\\
&\mspace{54mu}\cdot\expect_{\bm{\sigma}_{-i}}[\paymecgen_i(0,\Databf_{-i})\mid \Signalbf_{-i}=\signalbf_{-i}]\Bigr),\mspace{6mu}\signal_i\in\{0,1\},
\end{align*}
are the expected payment received by \player{}~$i$ when she reports $0$, weighted by $\Pr(\Signal_i=1)$ and $\Pr(\Signal_i=0)$ when her private signal is $1$ and $0$, respectively. Note that $K_1$, $K_0$, $L_1$ and $L_0$ do not depend on $p_1$, $p_0$, $q_1$, and $q_0$. 
The privacy level of the reported data at strategy $\sigma_i'$ is
\begin{equation*}
\zeta(\sigma_i')=\max\biggl\{\biggl|\ln\frac{p_1}{p_0}\biggr|,\biggl|\ln\frac{1-p_1}{1-p_0}\biggr|,\biggl|\ln\frac{q_1}{q_0}\biggr|,\biggl|\ln\frac{1-q_1}{1-q_0}\biggr|,
\biggl|\ln\frac{1-p_1-q_1}{1-p_0-q_0}\biggr|,\biggl|\ln\frac{p_1+q_1}{p_0+q_0}\biggr|\biggr\}.
\end{equation*}
With a little abuse of notation, we regard $\zeta(\sigma_i')$ as a function $\zeta(p_1,p_0,q_1,q_0)$. The expected utility of \player{}~$i$ can thus be written as
\begin{equation*}
\begin{split}
&\mspace{23mu}U_i(p_1,p_0,q_1,q_0)\\
&=\expect_{(\sigma_i',\bm{\sigma}_{-i})}[R_i(\Databf)-\cost(\zeta(\sigma_i'))]\\
&=K_1p_1+K_0p_0+L_1q_1+L_0q_0-\cost(\zeta(p_1,p_0,q_1,q_0)).
\end{split}
\end{equation*}

Next we discuss the best response of \player{}~$i$ for different cases of the values of $K_1$, $K_0$, $L_1$ and $L_0$. Since $\paymecgenbf$ is a nonnegative payment mechanism, these values are all nonnegative. Notice that for any $\signal_i\in\{0,1\},\signalbf_{-i}\in\{0,1\}^{N-1}$, $\Pr(\Signal_i=\signal_i,\Signalbf_{-i}=\signalbf_{-i})>0$. Therefore, $K_1$ and $K_0$ are either both equal to zero or both positive. The same argument also applies to $L_1$ and $L_0$. (1) When all of $K_1$, $K_0$, $L_1$ and $L_0$ are zero, a best response of \player{}~$i$ should minimize the privacy cost. Thus the strategy of \player{}~$i$ in a Nash equilibrium is to report $\Data_i$ that is independent of $\Signal_i$ so the privacy cost is zero. (2) When $K_1$ and $K_0$ are positive but $L_1$ and $L_0$ are zero, the best response of \player{}~$i$ is to always report $\Data_i=1$. (3) Similarly, when $K_1$ and $K_0$ are zero but $L_1$ and $L_0$ are positive, the best response of \player{}~$i$ is to always report $\Data_i=0$. We can see that the strategy of \player{}~$i$ in a Nash equilibrium is non-informative in all the three cases above. (4) In the remainder of this proof, we focus on the case that all of $K_1$, $K_0$, $L_1$ and $L_0$ are positive.

If a best response of \player{}~$i$ is to always not participate, then it is a non-informative strategy. Otherwise, a best response of \player{}~$i$ is specified by an optimal solution of the following optimization problem:
\begin{equation}\label{eq:opt}\tag{P}
\begin{split}
\underset{p_1,p_0,q_1,q_0}{\text{max}}\mspace{30mu} & U_i(p_1,p_0,q_1,q_0)\\
\text{subject to}\mspace{30mu}& 0\le p_1\le 1,0\le q_1 \le 1,\nonumber\\
& 0\le p_1+q_1\le 1,\nonumber\\
& 0\le p_0 \le 1,0\le q_0\le 1,\nonumber\\
& 0\le p_0+q_0\le 1,\nonumber\\
& p_1+q_1+p_0+q_0>0.\nonumber
\end{split}
\end{equation}

First, we prove that an optimal solution $(p_1^*,p_0^*,q_1^*,q_0^*)$ of \eqref{eq:opt} must satisfy that $p_1^*+q_1^*=p_0^*+q_0^*$. Suppose not. Without loss of generality we assume that $p_1^*+q_1^*<p_0^*+q_0^*$. We will find another solution $(p_1',p_0^*,q_1',q_0^*)$ that yields better utility, which contradicts the optimality of $(p_1^*,p_0^*,q_1^*,q_0^*)$.

\begin{sloppypar}
Since we assume that $p_1^*+q_1^*<p_0^*+q_0^*$, then at least one of the following two inequality holds: $p_1^*<p_0^*$, $q_1^*<q_0^*$. Still without loss of generality we assume that $p_1^*<p_0^*$. Then if $q_1^*< q_0^*$, let $p_1'=p_0^*$ and $q_1'=q_0^*$. Since $K_1$ and $L_1$ are positive, $(p_1',p_0^*,q_1',q_0^*)$ yields higher payment. It is easy to verify that $\zeta(p_1',p_0^*,q_1',q_0^*)<\zeta(p_1^*,p_0^*,q_1^*,q_0^*)$. Thus $(p_1',p_0^*,q_1',q_0^*)$ yields better utility. For the other case that $q_1^*\ge q_0^*$, let $p_1'=p_0^*+q_0^*-q_1^*$ and $q_1'=q_1^*$. Then $p_1^*<p_1'\le p_0^*$. Since $K_1$ is positive, $(p_1',p_0^*,q_1',q_0^*)$ yields higher payment. To check the privacy cost, notice that
\begin{align*}
\zeta(p_1^*,p_0^*,q_1^*,q_0^*)&=\max\biggl\{\ln\frac{p_0^*}{p_1^*},\ln\frac{1-p_1^*}{1-p_0^*},\ln\frac{q_1^*}{q_0^*},\ln\frac{1-q_0^*}{1-q_1^*},\\
&\mspace{72mu}\ln\frac{1-p_1^*-q_1^*}{1-p_0^*-q_0^*},\ln\frac{p_0^*+q_0^*}{p_1^*+q_1^*}\biggr\},
\end{align*}
and
\begin{equation*}
\zeta(p_1',p_0^*,q_1',q_0^*)=\max\biggl\{\ln\frac{p_0^*}{p_1'},\ln\frac{1-p_1'}{1-p_0^*},\ln\frac{q_1'}{q_0^*},\ln\frac{1-q_0^*}{1-q_1'}\biggr\}.
\end{equation*}
Since $p_1'>p_1^*$ and $q_1'=q_1^*$, $\zeta(p_1',p_0^*,q_1',q_0^*)\le \zeta(p_1^*,p_0^*,q_1^*,q_0^*)$. Thus $(p_1',p_0^*,q_1',q_0^*)$ yields better utility. Therefore, by contradiction, we must have $p_1^*+q_1^*=p_0^*+q_0^*$.
\end{sloppypar}

Next, we prove that an optimal solution $(p_1^*,p_0^*,q_1^*,q_0^*)$ must satisfy that $p_1^*+q_1^*=p_0^*+q_0^*=1$. Still, suppose not. Then we will find another solution $(p_1',p_0',q_1',q_0')$ that yields better utility. Let
\begin{align*}
p_1'=\frac{p_1^*}{p_1^*+q_1^*},\quad q_1'=\frac{q_1^*}{p_1^*+q_1^*},\\ p_0'=\frac{p_0^*}{p_0^*+q_0^*},\quad q_0'=\frac{q_0^*}{p_0^*+q_0^*}.
\end{align*}
We have seen that $p_1^*+q_1^*=p_0^*+q_0^*$. By the last constraint of \eqref{eq:opt}, $p_1^*+q_1^*=p_0^*+q_0^*>0$. Since we assume that $p_1^*+q_1^*$ and $p_0^*+q_0^*$ are not equal to $1$, they must be less than $1$. Since $K_1$, $K_0$, $L_1$ and $L_0$ are positive, $(p_1',p_0',q_1',q_0')$ yields higher payment. It is easy to verify that $\zeta(p_1',p_0',q_1',q_0')\le \zeta(p_1^*,p_0^*,q_1^*,q_0^*)$. Thus $(p_1',p_0',q_1',q_0')$ yields better utility, which contradicts the optimality of $(p_1^*,p_0^*,q_1^*,q_0^*)$.

By the results above, to find an optimal solution of \eqref{eq:opt}, we can focus on feasible $(p_1,p_0,q_1,q_0)$ such that $q_1=1-p_1$ and $q_0=1-p_0$. Let
\begin{equation*}
\overline{U}_i(p_1,p_0)=\overline{K}_1p_1+\overline{K}_0p_0+\overline{K}-\cost(\zeta(p_1,p_0)),
\end{equation*}
where $\overline{K}_1=K_1-L_1$, $\overline{K}_0=K_0-L_0$, $\overline{K}=L_1+L_0$, and with a little abuse of notation,
\begin{equation*}
\zeta(p_1,p_0)=\max\biggl\{\biggl|\ln\frac{p_1}{p_0}\biggr|,\biggl|\ln\frac{1-p_1}{1-p_0}\biggr|\biggr\}.
\end{equation*}
Then $(p_1^*,p_0^*,q_1^*,q_0^*)$ is an optimal solution of \eqref{eq:opt} if and only if $(p_1^*,p_0^*)$ is an optimal solution of the following optimization problem \ref{eq:opt'}:
\begin{equation}
\underset{0\le p_1\le 1,0\le p_0\le 1}{\text{max}}\quad\overline{U}_i(p_1,p_0)\label{eq:opt'}\tag{P'}
\end{equation}
Let $(p_1^*,p_0^*)$ be an optimal solution of \eqref{eq:opt'}. The strategy specified by $(p_1^*,p_0^*,1-p_1^*,1-p_0^*)$ is a symmetric randomized response if $p_1^*+p_0^*=1$, and is non-informative if $p_1^*=p_0^*$. Thus it suffices to prove that if $p_1^*+p_0^*\neq 1$, then $p_1^*=p_0^*$. We divide the case that $p_1^*+p_0^*\neq 1$ into two cases: $p_1^*+p_0^*>1$ and $p_1^*+p_0^*<1$, and prove that $p_1^*=p_0^*$ in both cases.

\noindent\textbf{Case 1:} $p_1^*+p_0^*>1$. Suppose, for contradiction, $p_1^*\neq p_0^*$.

\begin{sloppypar}
If $p_1^*=1$, then
\begin{equation*}
\max\biggl\{\biggl|\ln\frac{p_1^*}{p_0^*}\biggr|,\biggl|\ln\frac{1-p_1^*}{1-p_0^*}\biggr|\biggr\}=+\infty.
\end{equation*}
Consider $p_1=1$ and $p_0=1$. Then by the convention
\begin{equation*}
\max\biggl\{\biggl|\ln\frac{p_1}{p_0}\biggr|,\biggl|\ln\frac{1-p_1}{1-p_0}\biggr|\biggr\}=0.
\end{equation*}
Since $\overline{U}_i(p_1^*,p_0^*)\ge \overline{U}_i(p_1,p_0)$, then $\overline{K}_1+\overline{K}_0p_0^*-\cost(+\infty)\ge\overline{K}_1+\overline{K}_0-\cost(0)$. Thus
\begin{equation}\label{eq:cost-infty}
\cost(+\infty)\le-\overline{K}_0(1-p_0^*)< +\infty.
\end{equation}
Since $\cost(+\infty)\ge 0$, this also indicates that $\overline{K}_0\le 0$. Next consider $p_1=1$ and $p_0=0$. Then
\begin{equation*}
\max\biggl\{\biggl|\ln\frac{p_1}{p_0}\biggr|,\biggl|\ln\frac{1-p_1}{1-p_0}\biggr|\biggr\}=+\infty.
\end{equation*}
Since $\overline{U}_i(p_1^*,p_0^*)\ge\overline{U}_i(p_1,p_0)$, then $\overline{K}_1+\overline{K}_0p_0^*-\cost(+\infty)\ge\overline{K}_1-\cost(+\infty)$. Thus $\overline{K}_0\ge 0$, where we have used the fact that $\cost(+\infty)<+\infty$. Combining the above arguments we have $\overline{K}_0=0$. However, by \eqref{eq:cost-infty}, this indicates that $\cost(+\infty)=0$, which contradicts the assumption that $\cost(\xi)=0$ only for $\xi=0$. Therefore, $p_1^*\neq 1$. Following similar arguments we have $p_0^*\neq 1$, either.
\end{sloppypar}

\begin{sloppypar}
If $p_1^*>p_0^*$, then noticing that $p_1^*+p_0^*>1$ we have
\begin{equation*}
\max\biggl\{\biggl|\ln\frac{p_1^*}{p_0^*}\biggr|,\biggl|\ln\frac{1-p_1^*}{1-p_0^*}\biggr|\biggr\}=\ln\frac{1-p_0^*}{1-p_1^*}.
\end{equation*}
Consider
\begin{align*}
p_1&=\frac{\frac{1-p_0^*}{1-p_1^*}}{\frac{1-p_0^*}{1-p_1^*}+1},\quad p_0=\frac{1}{\frac{1-p_0^*}{1-p_1^*}+1}.
\end{align*}
Then
\begin{equation*}
\max\biggl\{\biggl|\ln\frac{p_1}{p_0}\biggr|,\biggl|\ln\frac{1-p_1}{1-p_0}\biggr|\biggr\}=\ln\frac{1-p_0^*}{1-p_1^*}.
\end{equation*}
Since $\overline{U}_i(p_1^*,p_0^*)\ge \overline{U}_i(p_1,p_0)$, then
\begin{align*}
&\mspace{23mu}\overline{K}_1p_1^*+\overline{K}_0p_0^*-\cost\Biggl(\ln\frac{1-p_0^*}{1-p_1^*}\Biggr)+\overline{K}\\
&\ge\overline{K}_1p_1+\overline{K}_0p_0-\cost\Biggl(\ln\frac{1-p_0^*}{1-p_1^*}\Biggr)+\overline{K}.
\end{align*}
Thus, inserting $p_1$ and $p_0$ we obtain
\begin{equation*}
\overline{K}_1(1-p_1^*)+\overline{K}_0(1-p_0^*)\ge 0,
\end{equation*}
where we have used the condition $p_1^*+p_0^*>1$. Next still consider $p_1=p_0=1$. Since $\overline{U}_i(p_1^*,p_0^*)\ge \overline{U}_i(p_1,p_0)$, then
\begin{equation*}
\overline{K}_1p_1^*+\overline{K}_0p_0^*-\cost\Biggl(\ln\frac{1-p_0^*}{1-p_1^*}\Biggr)\ge \overline{K}_1+\overline{K}_0-\cost(0).
\end{equation*}
Thus
\begin{equation*}
-\cost\Biggl(\ln\frac{1-p_0^*}{1-p_1^*}\Biggr)\ge\overline{K}_1(1-p_1^*)+\overline{K}_0(1-p_0^*)\ge 0.
\end{equation*}
which indicates that
\begin{equation*}
\ln\frac{1-p_0^*}{1-p_1^*}=0.
\end{equation*}
Therefore, $p_1^*=p_0^*$, which contradicts the assumption.
\end{sloppypar}

\begin{sloppypar}
If $p_1^*<p_0^*$, then noticing that $p_1^*+p_0^*>1$ we have
\begin{equation*}
\max\biggl\{\biggl|\ln\frac{p_1^*}{p_0^*}\biggr|,\biggl|\ln\frac{1-p_1^*}{1-p_0^*}\biggr|\biggr\}=\ln\frac{1-p_1^*}{1-p_0^*}.
\end{equation*}
We use similar arguments to obtain contradiction. Consider
\begin{align*}
p_1=\frac{1}{\frac{1-p_1^*}{1-p_0^*}+1},\quad p_0=\frac{\frac{1-p_1^*}{1-p_0^*}}{\frac{1-p_1^*}{1-p_0^*}+1}.
\end{align*}
Then since $\overline{U}_i(p_1^*,p_0^*)\ge \overline{U}_i(p_1,p_0)$, we have $\overline{K}_1(1-p_1^*)+\overline{K}_0(1-p_0^*)\ge 0$. Next still consider $p_1=p_0=1$. Then since $\overline{U}_i(p_1^*,p_0^*)\ge \overline{U}_i(p_1,p_0)$, we have
\begin{equation*}
-\cost\Biggl(\ln\frac{1-p_1^*}{1-p_0^*}\Biggr)\ge\overline{K}_1(1-p_1^*)+\overline{K}_0(1-p_0^*)\ge 0.
\end{equation*}
which again indicates that
\begin{equation*}
\ln\frac{1-p_1^*}{1-p_0^*}=0.
\end{equation*}
Therefore, $p_1^*=p_0^*$, which contradicts the assumption.
\end{sloppypar}

In summary, for the case that $p_1^*+p_0^*>1$, $p_1^*=p_0^*$.

\noindent\textbf{Case 2:} $p_1^*+p_0^*<1$. Suppose, for contradiction, $p_1^*\neq p_0^*$. Then we obtain contradictions by similar arguments as used in Case~1.

\begin{sloppypar}
First, by comparing $\overline{U}_i(p_1^*,p_0^*)$ with $\overline{U}_i(0,0)$ and $\overline{U}_i(0,1)$ we can prove that $p_1^*\neq 0, p_0^*\neq 0$. If $p_1^*>p_0^*$, then by comparing $\overline{U}_i(p_1^*,p_0^*)$ with the expected utility at
\begin{equation*}
p_1=\frac{\frac{p_1^*}{p_0^*}}{\frac{p_1^*}{p_0^*}+1},\quad p_0=\frac{1}{\frac{p_1^*}{p_0^*}+1},
\end{equation*}
we have $\overline{K}_1p_1^*+\overline{K}_0p_0^*\le 0$. By comparing $\overline{U}_i(p_1^*,p_0^*)$ with $\overline{U}_i(0,0)$, we have
\begin{equation*}
\cost\Biggl(\ln\frac{1-p_1^*}{1-p_0^*}\Biggr)\le\overline{K}_1p_1^*+\overline{K}_0p_0^*\le 0.
\end{equation*}
Therefore, $p_1^*=p_0^*$, which contradicts the assumption. If $p_1^*<p_0^*$, then by comparing $\overline{U}_i(p_1^*,p_0^*)$ with the expected utility at
\begin{equation*}
p_1=\frac{1}{\frac{p_0^*}{p_1^*}+1},\quad p_0=\frac{\frac{p_0^*}{p_1^*}}{\frac{p_0^*}{p_1^*}+1},
\end{equation*}
we have $\overline{K}_1p_1^*+\overline{K}_0p_0^*\le 0$. By comparing $\overline{U}_i(p_1^*,p_0^*)$ with $\overline{U}_i(0,0)$, we have
\begin{equation*}
\cost\Biggl(\ln\frac{1-p_1^*}{1-p_0^*}\Biggr)\le\overline{K}_1p_1^*+\overline{K}_0p_0^*\le 0.
\end{equation*}
Therefore, $p_1^*=p_0^*$, which contradicts the assumption. In summary, for the case that $p_1^*+p_0^*<1$, we also have $p_1^*=p_0^*$ by similar arguments as used in Case~1. This completes the proof.
\end{sloppypar}
\end{proof}

\def\lemNoPositive{\ref{lem:positive-eps}}
\section{Proof of Lemma~\lemNoPositive}\label{app:lem:positive-eps}
\begin{proof}
For any nonnegative payment mechanism $\paymecgenbf$ in which the strategy profile $(\sigma_i^{(-\epsilon)},\bm{\sigma}_{-i})$ is a Nash equilibrium, consider the payment mechanism $\paymecgenbf'$ defined by
\begin{equation*}
\paymecgenbf'(\data_i,\databf_{-i})=\paymecgenbf(1-\data_i,\databf_{-i}).
\end{equation*}

We first prove that $(\sigma_i^{(\epsilon)},\bm{\sigma}_{-i})$ is a Nash equilibrium in $\paymecgenbf'$. For an \player{}~$i$, consider any strategy $\sigma_i'$ of \player{}~$i$ and let
\begin{equation*}
\begin{split}
p_1=\Pr_{\sigma_i'}(\Data_i=1\mid \Signal_i=1),\quad& q_1=\Pr_{\sigma_i'}(\Data_i=0\mid \Signal_i=1),\\
p_0=\Pr_{\sigma_i'}(\Data_i=1\mid \Signal_i=0),\quad& q_0=\Pr_{\sigma_i'}(\Data_i=0\mid \Signal_i=0).
\end{split}
\end{equation*}
We say $(p_1,p_0,q_1,q_0)$ is feasible if it satisfies that
\begin{gather*}
0\le p_1\le 1,\quad 0\le q_1\le 1,\quad 0\le p_1+q_1\le 1,\\
0\le p_0\le 1,\quad 0\le q_0\le 1,\quad 0\le p_0+q_0\le 1.
\end{gather*}
Then following the notation in the proof of Lemma~\ref{lem:equi-char}, in the mechanism $\paymecgenbf'$ and $\paymecgenbf$, we denote the expected utility of \player{}~$i$ at $\sigma_i'$ when other \players{} follow $\bm{\sigma}_{-i}$ by $U_i'(p_1,p_0,q_1,q_0)$ and $U_i(p_1,p_0,q_1,q_0)$, respectively, and they can be written as follows:
\begin{align*}
U_i'(p_1,p_0,q_1,q_0)&=K'_{1,i}p_1+K'_{0,i}p_0+L'_{1,i}q_1+L'_{0,i}q_0\\
&\mspace{23mu}-\cost(\zeta(p_1,p_0,q_1,q_0)),\\
U_i(p_1,p_0,q_1,q_0)&=K_{1,i}p_1+K_{0,i}p_0+L_{1,i}q_1+L_{0,i}q_0\\
&\mspace{23mu}-\cost(\zeta(p_1,p_0,q_1,q_0)).
\end{align*}
We derive the relations between $K'_{1,i}$, $K'_{0,i}$, $L'_{1,i}$, $L'_{0,i}$ and $K_{1,i}$, $K_{0,i}$, $L_{1,i}$, $L_{0,i}$. By definition,
\begin{align*}
K'_{1,i}&=\sum_{\databf_{-i}}\paymecgen'_i(1,\databf_{-i})\Pr_{\bm{\sigma}_{-i}}(\Databf_{-i}=\databf_{-i},\Signal_i=1)\\
&=\sum_{\databf_{-i}}\paymecgen_i(0,\databf_{-i})\Pr_{\bm{\sigma}_{-i}}(\Databf_{-i}=\databf_{-i},\Signal_i=1)\\
&=L_{1,i}.
\end{align*}
Similarly, $K_{0,i}'=L_{0,i}$, $L_{1,i}'=K_{1,i}$ and $L_{0,i}'=K_{0,i}$. Since $(\sigma_i^{(-\epsilon)},\bm{\sigma}_{-i})$ is a Nash equilibrium in $\paymecgenbf$, for any feasible $(p_1,p_0,q_1,q_0)$,
\begin{equation*}
U_i\biggl(\frac{1}{e^{\epsilon}+1},\frac{e^{\epsilon}}{e^{\epsilon}+1},\frac{e^{\epsilon}}{e^{\epsilon}+1},\frac{1}{e^{\epsilon}+1}\biggr)\ge U_i(p_1,p_0,q_1,q_0).
\end{equation*}
Therefore, for any feasible $(p_1,p_0,q_1,q_0)$,
\begin{equation*}
U'_i\biggl(\frac{e^{\epsilon}}{e^{\epsilon}+1},\frac{1}{e^{\epsilon}+1},\frac{1}{e^{\epsilon}+1},\frac{e^{\epsilon}}{e^{\epsilon}+1}\biggr)\ge U'_i(p_1,p_0,q_1,q_0),
\end{equation*}
where we have used the symmetry property of the cost function $\cost$. This implies that $\sigma_i^{(\epsilon)}$ is a best response of \player{}~$i$ in $\paymecgenbf'$ when other \players{} follow $\bm{\sigma}_{-i}$. Now consider any \player{}~$j$ with $j\neq i$ and any strategy $\sigma_j'$. Let
\begin{equation*}
\begin{split}
p_1=\Pr_{\sigma_j'}(\Data_j=1\mid S_j=1),\quad & q_1=\Pr_{\sigma_j'}(\Data_j=0\mid S_j=1),\\
p_0=\Pr_{\sigma_j'}(\Data_j=1\mid S_j=0),\quad& q_0=\Pr_{\sigma_j'}(\Data_j=0\mid S_j=0).
\end{split}
\end{equation*}
Let $\sigma_{-j}^{(\epsilon)}=(\sigma_i^{(\epsilon)},\bm{\sigma}_{-i,j})$ and $\sigma_{-j}^{(-\epsilon)}=(\sigma_i^{(-\epsilon)},\bm{\sigma}_{-i,j})$. Then similarly, in the mechanism $\paymecgenbf'$ and $\paymecgenbf$, we denote the expected utility of \player{}~$j$ at $\sigma_j'$ when other \players{} follow $\sigma_{-j}^{(\epsilon)}$ and $\sigma_{-j}^{(-\epsilon)}$ by $U'_j(p_1,p_0,q_1,q_0)$ and $U_j(p_1,p_0,q_1,q_0)$, respectively, and they can be written as follows:
\begin{align*}
U'_j(p_1,p_0,q_1,q_0)&=K'_{1,j}p_1+K'_{0,j}p_0+L'_{1,j}q_1+L'_{0,j}q_0\\
&\mspace{23mu}-\cost(\zeta(p_1,p_0,q_1,q_0)),\\
U_j(p_1,p_0,q_1,q_0)&=K_{1,j}p_1+K_{0,j}p_0+L_{1,j}q_1+L_{0,j}q_0\\
&\mspace{23mu}-\cost(\zeta(p_1,p_0,q_1,q_0)).
\end{align*}
We derive the relations between $K'_{1,j}$, $K'_{0,j}$, $L'_{1,j}$, $L'_{0,j}$ and $K_{1,j}$, $K_{0,j}$, $L_{1,j}$, $L_{0,j}$. By definition,
\begin{align*}
K'_{1,j}&=\sum_{\databf_{-i,j}}\sum_{\data_i}\paymecgen'_i(\data_i,1,\databf_{-i,j})\\
&\mspace{23mu}\cdot\sum_{\signal_i}\Pr_{\sigma_i^{(\epsilon)}}(\Data_i=\data_i\mid \Signal_i=\signal_i)\\
&\mspace{60mu}\cdot\Pr_{\bm{\sigma}_{-i,j}}(\Databf_{-i,j}=\databf_{-i,j},\Signal_i=\signal_i,S_j=1)\\
&=\sum_{\databf_{-i,j}}\sum_{\data_i}\paymecgen_i(1-\data_i,1,\databf_{-i,j})\\
&\mspace{23mu}\cdot\sum_{\signal_i}\Pr_{\sigma_i^{(-\epsilon)}}(\Data_i=1-\data_i\mid \Signal_i=\signal_i)\\
&\mspace{60mu}\cdot\Pr_{\bm{\sigma}_{-i,j}}(\Databf_{-i,j}=\databf_{-i,j},\Signal_i=\signal_i,S_j=1)\\
&=K_{1,j}.
\end{align*}
Similarly, $K'_{0,j}=K_{0,j}$, $L'_{1,j}=L_{1,j}$, and $L'_{0,j}=L_{0,j}$. Therefore, for any feasible $(p_1,p_0,q_1,q_0)$,
\begin{equation*}
U'_j(p_1,p_0,q_1,q_0)=U_j(p_1,p_0,q_1,q_0).
\end{equation*}
Thus $\sigma_j$ is a best response of \player{}~$j$ in $\paymecgenbf'$ when other \players{} follow $\sigma_{-j}^{(\epsilon)}$. This completes the proof for $(\sigma_i^{(\epsilon)},\bm{\sigma}_{-i})$ to be a Nash equilibrium in $\paymecgenbf'$.

With the above proof, it is not hard to verify that the expected payment to each \player{} at these two equilibria of the two mechanisms are the same.
\end{proof}

\def\lemNoGenie{\ref{lem:genie}}
\section{Proof of Lemma~\lemNoGenie}\label{app:lem:genie}
\begin{proof}
Consider any payment mechanism $\paymecgenbf$ and any Nash equilibrium $\bm{\sigma}$ of it. We will construct a genie-aided mechanism $\widehat{\paymecgenbf}$ such that $\bm{\sigma}$ is also a Nash equilibrium of $\widehat{\paymecgenbf}$ and the expected payment to each \player{} at this equilibrium is the same under $\paymecgenbf$ and $\widehat{\paymecgenbf}$.

As in the proof of Lemma~\ref{lem:equi-char}, for any \player{}~$i$, consider any strategy $\sigma_i'$ of \player{}~$i$ and let
\begin{equation*}
\begin{split}
p_1=\Pr_{\sigma_i'}(\Data_i=1\mid \Signal_i=1),\quad& q_1=\Pr_{\sigma_i'}(\Data_i=0\mid \Signal_i=1),\\
p_0=\Pr_{\sigma_i'}(\Data_i=1\mid \Signal_i=0),\quad& q_0=\Pr_{\sigma_i'}(\Data_i=0\mid \Signal_i=0).
\end{split}
\end{equation*}
Then we will first derive the expected utility of \player{}~$i$ at the strategy $\sigma_i'$ as a function of $(p_1,p_0,q_1,q_0)$, denoted by $U_i(p_1,p_0,q_1,q_0)$, but using a slightly different expression from the form in Lemma~\ref{lem:equi-char}. When other \players{} follow $\bm{\sigma}_{-i}$, the expected payment to \player{}~$i$ at the strategy $\sigma_i'$ can be written as
\begin{align*}
&\mspace{23mu}\expect_{(\sigma_i',\bm{\sigma}_{-i})}[\paymecgen_i(\Databf)]\\
&=\sum_{\data_i,\signal_i,\state}\biggl\{\Pr_{\sigma_i'}(\Data_i=\data_i,\Signal_i=\signal_i,\State=\state)\\
&\mspace{81mu}\cdot\expect_{(\sigma_i',\bm{\sigma}_{-i})}[\paymecgen_i(\Databf)\mid \Data_i=\data_i, \Signal_i=\signal_i,\State=\state]\biggr\}\\
&=\sum_{\data_i,\signal_i,\state}\biggl\{\Pr_{\sigma_i'}(\Data_i=\data_i\mid \Signal_i=\signal_i)\Pr(\Signal_i=\signal_i,\State=\state)\\
&\mspace{81mu}\cdot\expect_{(\sigma_i',\bm{\sigma}_{-i})}[\paymecgen_i(\Databf)\mid \Data_i=\data_i, \State=\state]\biggr\},
\end{align*}
where we have used the fact that $\Data_i$ is independent from $\State$ given $\Signal_i$, and $\Databf_{-i}$ is independent from $\Signal_i$ given $\Data_i$ and $\State$. Let $\overline{\paymecgen}_i(\data_i;\state)$ denote $\expect_{(\sigma_i',\bm{\sigma}_{-i})}[\paymecgen_i(\Databf)\mid \Data_i=\data_i, \State=\state]$ for $\data_i,\state\in\{0,1\}$. Then $\overline{\paymecgen}_i(\data_i;\state)$ does not depend on the strategy of \player{}~$i$ since
\begin{align*}
&\mspace{23mu}\expect_{(\sigma_i',\bm{\sigma}_{-i})}[\paymecgen_i(\Databf)\mid \Data_i=\data_i, \State=\state]\\
&=\expect_{(\sigma_i',\bm{\sigma}_{-i})}[\paymecgen_i(\data_i,\Databf_{-i})\mid \Data_i=\data_i, \State=\state]\\
&=\expect_{\bm{\sigma}_{-i}}[\paymecgen_i(\data_i,\Databf_{-i})\mid \State=\state],
\end{align*}
where the last equality follows from the conditional independence between $\Data_i$ and $\Databf_{-i}$ given $\State$. With this notation,
\begin{align*}
&\mspace{23mu}\expect_{(\sigma_i',\bm{\sigma}_{-i})}[\paymecgen_i(\Databf)]\\
&=\sum_{\data_i,\signal_i}\biggl\{\Pr_{\sigma_i'}(\Data_i=\data_i\mid \Signal_i=\signal_i)\\
&\mspace{66mu}\cdot\sum_{\state}\Pr(\Signal_i=\signal_i,\State=\state)\overline{\paymecgen}_i(\data_i;\state)\biggr\}.
\end{align*}
Therefore, the expected utility of \player{}~$i$ is given by
\begin{align*}
&\mspace{23mu}U_i(p_1,p_0,q_1,q_0)\\
&=\expect_{(\sigma_i',\bm{\sigma}_{-i})}[R_i(\Databf)-\cost(\zeta(\sigma_i'))]\\
&=K_1p_1+K_0p_0+L_1q_1+L_0q_0-\cost(\zeta(p_1,p_0,q_1,q_0)),
\end{align*}
where
\begin{align*}
K_{\signal_i}=\sum_{\state}\Pr(\Signal_i=\signal_i,\State=\state)\overline{\paymecgen}_i(1;\state),\mspace{6mu} \signal_i\in\{0,1\},\\
L_{\signal_i}=\sum_{\state}\Pr(\Signal_i=\signal_i,\State=\state)\overline{\paymecgen}_i(0;\state),\mspace{6mu} \signal_i\in\{0,1\}.
\end{align*}

Consider a genie-aided mechanism $\widehat{\paymecgenbf}$ defined as follows: for any \player{}~$i$,
\begin{align*}
\widehat{\paymecgen}_i(\data_i,\state)=\overline{\paymecgen}_i(\data_i;\state),\mspace{6mu} \data_i\in\mathcal{\Data},\state\in\{0,1\}.
\end{align*}
Still consider any \player{}~$i$ and any strategy $\sigma_i'$ of \player{}~$i$. Let $\widehat{U}_i(p_1,p_0,q_1,q_0)$ denote the expected utility of \player{}~$i$ at the strategy $\sigma_i'$ when other \players{} follow $\bm{\sigma}_{-i}$. Then
\begin{align*}
&\mspace{23mu}\widehat{U}_i(p_1,p_0,q_1,q_0)\\
&=\expect_{(\sigma_i',\bm{\sigma}_{-i})}\Bigl[\widehat{\paymecgen}_i(\Data_i,\State)-\cost(\zeta(p_1,p_0,q_1,q_0))\Bigr]\\
&=\sum_{\data_i,\signal_i,\state}\Pr_{\sigma_i'}(\Data_i=\data_i\mid \Signal_i=\signal_i)\Pr(\Signal_i=\signal_i,\State=\state)\widehat{\paymecgen}_i(\data_i,\state)\\
&\mspace{23mu}-\cost(\zeta(p_1,p_0,q_1,q_0))\\
&=K_1p_1+K_0p_0+L_1q_1+L_0q_0-\cost(\zeta(p_1,p_0,q_1,q_0)),
\end{align*}
where the last equality follows from the definition of $\widehat{\paymecgen}_i(\data_i,\state)$. Thus, $\widehat{U}_i(p_1,p_0,q_1,q_0)=U_i(p_1,p_0,q_1,q_0)$. Since $\bm{\sigma}$ is a Nash equilibrium of $\paymecgenbf$, the $(p_1,p_0,q_1,q_0)$ that corresponds to $\sigma_i$ maximizes $U_i(p_1,p_0,q_1,q_0)$, which implies that $\sigma_i$ is also a best response of \player{}~$i$ under the genie-aided mechanism $\widehat{\paymecgenbf}$ when other \players{} follow $\bm{\sigma}_{-i}$. Therefore, $\bm{\sigma}$ is also a Nash equilibrium of $\widehat{\paymecgenbf}$, and the expected payment to each \player{} at this equilibrium is the same under $\paymecgenbf$ and $\widehat{\paymecgenbf}$.
\end{proof}

\def\lemNoLower{\ref{lem:lower-P2}}
\section{Proof of Lemma~\lemNoLower}\label{app:proof-lower-P2}
\begin{proof}
We first prove that $\widetilde{\epsilon}$ is well-defined. Let a function $r\colon(0,+\infty)\rightarrow\mathbb{R}$ be defined as
\begin{equation*}
r(\epsilon)=\frac{D(\epsilon)}{\Value_{\mathrm{LB}}(\epsilon)}.
\end{equation*}
Let $P_1^{(\epsilon)}$ and $P_0^{(\epsilon)}$ be the conditional distributions of the reported $\Data_i$ at the $\epsilon$-strategy given $\State=1$ and $\State=0$, respectively, and let $P_{\mathrm{U}}$ be the uniform distribution on $\{0,1\}$. Then note that
\begin{align*}
D(\epsilon)&=D_{\mathrm{KL}}(P_{\mathrm{U}}||P_1^{(\epsilon)})=D_{\mathrm{KL}}(P_{\mathrm{U}}||P_0^{(\epsilon)})\\
&=\frac{1}{2}\ln\frac{(e^{\epsilon}+1)^2}{4(\quality e^{\epsilon}+1-\quality)((1-\quality)e^{\epsilon}+\quality)}.
\end{align*}
Therefore, the function $r$ is continuous on $(0,+\infty)$. Further, the function $r$ attains its maximum value in a bounded subset of $(0,+\infty)$ since for any $\epsilon\in(0,+\infty)$, $r(\epsilon)>0$, and
\begin{align*}
\lim_{\epsilon\rightarrow 0}r(\epsilon)&=0\\
\lim_{\epsilon\rightarrow +\infty}r(\epsilon)&=0.
\end{align*}
The set $\arg\max r(\epsilon)$ is a closed set since it is the inverse image of one point. Therefore, $\widetilde{\epsilon}=\inf\{\arg\max r(\epsilon)\}$ is well-defined.

Now consider any feasible $(N,\epsilon_1,\epsilon_2,\dots,\epsilon_N)$ of \eqref{eq:payment-accuracy-lower}. By the construction of $\widetilde{\epsilon}$, for any \player{}~$i$,
\begin{align*}
\Value_{\mathrm{LB}}(\epsilon_i)\ge \frac{\Value_{\mathrm{LB}}(\widetilde{\epsilon})}{D(\widetilde{\epsilon})} D(\epsilon_i).
\end{align*}
Then
\begin{align*}
\sum_{i=1}^N\Value_{\mathrm{LB}}(\epsilon_i)&\ge \frac{\Value_{\mathrm{LB}}(\widetilde{\epsilon})}{D(\widetilde{\epsilon})}\sum_{i=1}^ND(\epsilon_i)\\
&\ge \frac{\Value_{\mathrm{LB}}(\widetilde{\epsilon})}{D(\widetilde{\epsilon})}\ln(1/\err),
\end{align*}
where the second inequality follows from the feasibility of $(N,\epsilon_1,\epsilon_2,\dots,\epsilon_N)$. By the construction of $\widetilde{N}$,
\begin{equation*}
\widetilde{N}<\frac{\ln(1/\err)}{D(\widetilde{\epsilon})}+1.
\end{equation*}
Therefore,
\begin{align*}
\sum_{i=1}^N\Value_{\mathrm{LB}}(\epsilon_i)\ge (\widetilde{N}-1)\Value_{\mathrm{LB}}(\widetilde{\epsilon}),
\end{align*}
which completes the proof.
\end{proof}

\end{document}